\documentclass{svmult}
  
\usepackage{mathptmx,hyperref}   
\usepackage{helvet} 
\usepackage{courier}
\usepackage{amsmath,amsfonts}
\usepackage{graphicx}
\usepackage[bottom]{footmisc}          
    
\usepackage{mathptmx}         
\usepackage{helvet}            
\usepackage{courier}            
\usepackage{type1cm}             
\usepackage{makeidx}             
\usepackage{graphicx}          
\usepackage{multicol}          
\usepackage[bottom]{footmisc}  
     
\usepackage{epsfig}     
\usepackage{psfrag,rotating}     
\usepackage{amssymb,floatflt,enumerate} 
\usepackage{amsmath,amscd,psfrag,leqno} 
\usepackage[mathscr]{eucal}  
\usepackage{shadethm}           
\usepackage{graphicx,subfigure}   
\usepackage{overpic,contour}       
\contourlength{0.3mm}  
 
\def\Beweisende{\square}            
\def\BewEnde{\hfill{\Beweisende}}

\def\RR{{\mathbb R}}

\def\dach#1{\widehat{#1}}

\def\Vkt#1{{\mathbf #1}}

\newcommand{\go}[1]{{\sf #1}}

\newtheorem{convention}{Convention}
\makeindex
\begin{document}

\title*{%Tubular Bennett Structures
Flexible arrangement of two Bennett tubes
}

\author{Georg Nawratil}
\authorrunning{G. Nawratil}
%\institute{$^1$IRCCyN, CNRS, France, \email{Philippe.Wenger@irccyn.ec-nantes.fr} \\
%$^1$University of Minho, Portugal, \email{pflores@dem.uminho.pt}}
\institute{
  Institute of Discrete Mathematics and Geometry \&  
	Center for Geometry and Computational Design, TU Wien, \\
	\email{nawratil@geometrie.tuwien.ac.at}}

\maketitle

\abstract{
In analogy to flexible bipyramids, also known as Bricard octahedra, we study flexible couplings of two Bennett mechanisms. The resulting flexible bi-Bennett structures can be used as building blocks of flexible tubes with quadrilateral cross-section that possess skew faces.  
We present three 4-dimensional families of flexible arranged Bennett tubes and discuss some of their properties as well as their relation to flexible bipyramids and biprisms, respectively. 
}


\section{Introduction}\label{sec:intro}

In 1897 Bricard \cite{bricard} proved that there are three types of flexible octahedra
in the Euclidean 3-space. 
These so-called {\it Bricard octahedra}  can also be seen as flexible bipyramids, where each quadrilateral pyramid corresponds to a spherical 4R-loop. 
As also planar 4R-loops are flexible, we can replace one or both pyramids by quadrilateral prisms. 
The complete classification of flexible arrangements of two quadrilateral prisms was given in \cite{naw2} and of a quadrilateral pyramid and a quadrilateral prism in \cite{naw1}. 
These different types can also be combined to construct flexible tubular structures (see \cite{kiumars} and the references given therein as well as \cite{nelson}).

Beside planar and spherical 4R-loops there also exist spatial ones \cite{delassus}, known as Bennett mechanisms \cite{bennett}, which can be realized as so-called Bennett tubes by using skew faces (see Fig.\ \ref{fig1}).
Therefore, one can ask for flexible arrangements of a Bennett tube with a quadrilateral pyramid/prism and of two Bennett tubes, respectively.
We present results on the latter case within this paper, which is structured as follows: 

Section \ref{sec:bennett} is devoted to Bennett loops; starting with their basic geometric constraints, followed by the discussion of the planar and spherical limits of these structures (Subsections \ref{sec:planar} and \ref{sec:spherical}) and the Denavit-Hartenberg procedure (Subsection \ref{sec:dhproc}). This section is closed by pointing out some remarkable properties and special cases of the Bennett mechanisms (Subsection \ref{sec:prop}). 
In Section \ref{sec:pairedB}, we study the flexible coupling of two Bennett tubes, where we distinguish between symmetric arrangements (Subsection \ref{subsec:sym}) and general ones (Subsection \ref{subsec:nonsym}). The prismatic and conical limits of the obtained families are discussed in Section \ref{sec:limits_results}. 
Finally, the paper is concluded in Section 
\ref{sec:conclusion}, where also some open problems are stated.

 But before we plunge into medias res, we give a review on flexible arrangements of quadrilateral prisms/pyramids.

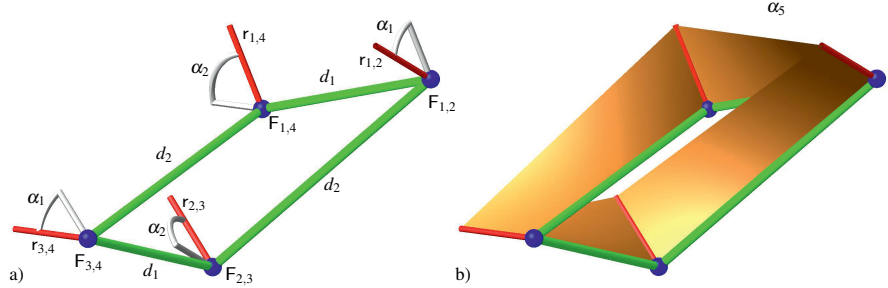
\begin{figure}[t]
\begin{center}
\begin{overpic}
   [height=35mm]{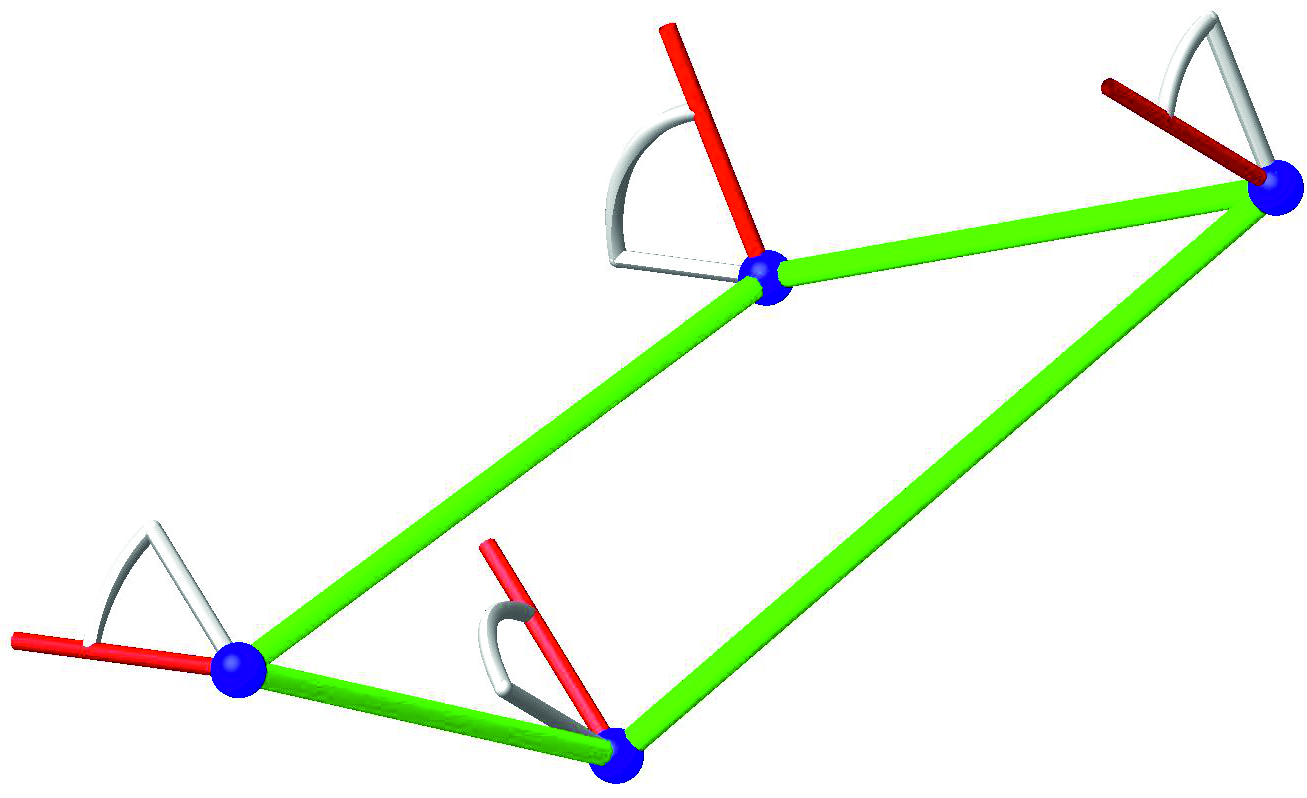}
\begin{scriptsize}
\put(0,0){a)}
\put(5,7){$\go r_{3,4}$}
\put(15,3){$\go F_{3,4}$}
\put(4,18){$\alpha_1$}
\put(40,17){$\go r_{2,3}$}
\put(50,0){$\go F_{2,3}$}
\put(32,11){$\alpha_2$}
\put(97,40){$\go F_{1,2}$}
\put(81,50){$\go r_{1,2}$}
\put(42,48){$\alpha_2$}
\put(54,56){$\go r_{1,4}$}
\put(60,35){$\go F_{1,4}$}
\put(85,58){$\alpha_1$}
\put(31,1){$d_1$}
\put(73,21){$d_2$}
\put(34,28){$d_2$}
\put(72,46){$d_1$}
\end{scriptsize}     
  \end{overpic} 
\,
\begin{overpic}
    [height=35mm]{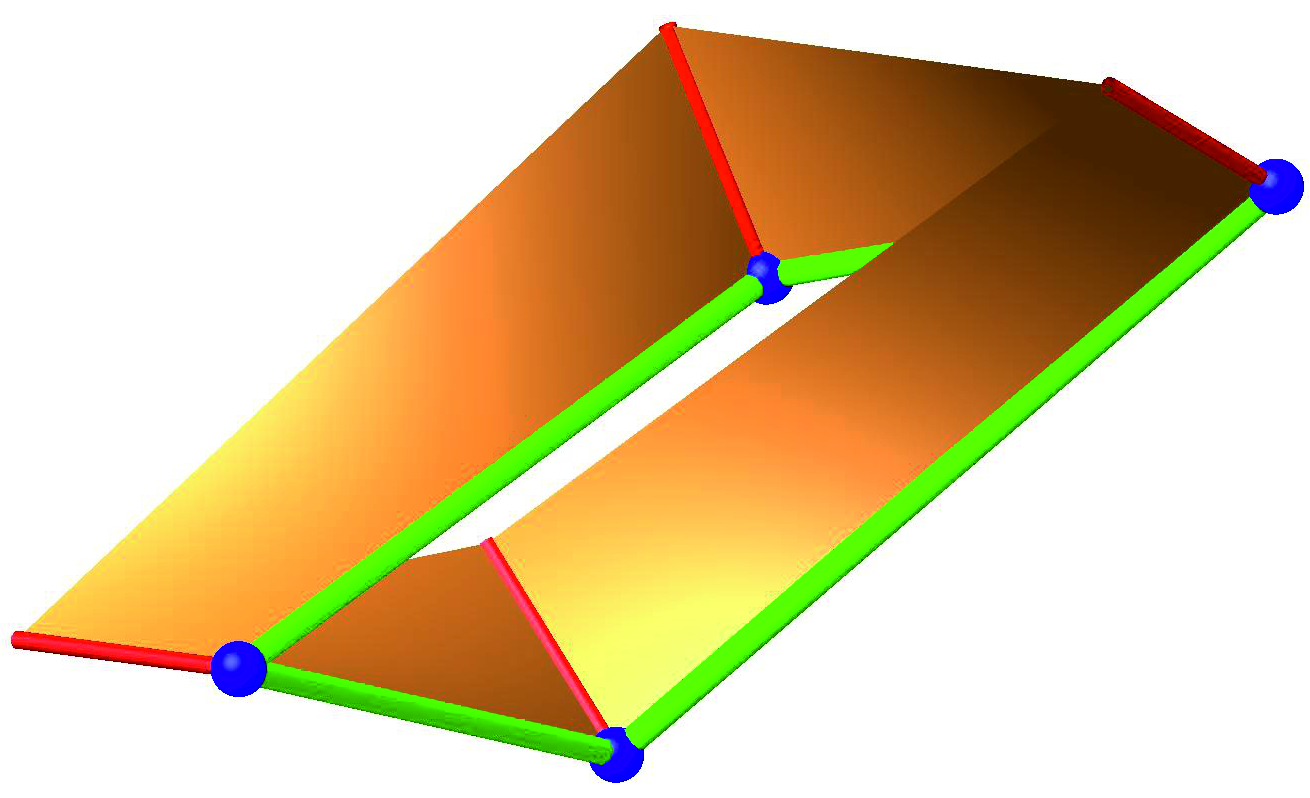}
\begin{scriptsize}
\put(0,0){b)}
\put(72,62){$\alpha_5$}
\end{scriptsize}     
  \end{overpic} 
\end{center}	
\caption{(a) Bennett loop including the used notation. (b) Self-intersection-free realization as Bennett tube, where the four skew quads are filled with hyperbolic paraboloid (HP) patches.}
  \label{fig1}
\end{figure}

\subsection{Review on flexible coupled quadrilateral prisms/pyramids}\label{sec:review}

\begin{enumerate}[(I)]
\item
{\bf Flexible quadrilateral bipyramids} \newline
According to Bricard \cite{bricard} there are three types of flexible octahedra\footnote{No face degenerates into a line (or even a point) 
and no two neighboring faces coincide permanently during the flex.}
in the Euclidean 3-space $E^3$. 
These so-called {\it Bricard octahedra} are as follows: 
\begin{enumerate}[Type 1]
\item 
All three pairs of opposite vertices are symmetric with respect to a common line.
\item
Two pairs of opposite vertices are symmetric with respect to a common plane $\varepsilon$ which passes through the remaining two vertices.  
\item
Two pairs of opposite vertices are V-hedral\footnote{The terms V-hedral and anti-V-hedral are also recapped within Definition
\ref{def:voss} given later on.} vertices and the remaining pair is an anti-V-hedral one (see also \cite[§ 13]{kokotsakis} or \cite{horns}). 
\end{enumerate}
Note that the latter case possesses two flat poses due to its 
composition from only V-hedral and anti-V-hedral vertices. 
For a detailed discussion of this type we also refer to \cite{stachel}. 
\item
{\bf Flexible arrangements of quadrilateral pyramids and prisms} \newline
These flexible couplings were studied in 
\cite{naw1} where it was shown that only two cases can exist, which are limits of Bricard octahedra of Type 2 (one vertex located in $\varepsilon$ is an ideal point) and Type 3 (one of the V-hedral vertices is an ideal point; see also \cite[Sec.\ 5]{stachel}).  
\item
{\bf Flexible quadrilateral biprisms} \newline
The complete classification of flexible arrangements of two quadrilateral prisms was given in \cite{naw2}, which reads as follows:
\begin{enumerate}[Type 1]
\item
The two pairs of opposite vertices are symmetric with respect to a common line as well as the edges of the prisms. 
\item
	\begin{enumerate}[i.]
	\item
	One pair of opposite vertices is symmetric with respect to a plane $\varepsilon$ which 
	contains the other pair of opposite vertices. Moreover, also the edges of the prisms	are symmetric with respect to $\varepsilon$. 
	\item
	The four vertices  are coplanar and form an anti-parallelogram and its plane of symmetry is parallel to the edges of the prisms.
	\end{enumerate}
\item
This type is characterized by the existence of two flat poses and consists of two prisms where the orthogonal cross-sections are congruent anti-parallelograms. 
\item
The four vertices are coplanar and form
	\begin{enumerate}[i.]
	\item
	a deltoid and the edges of the prisms are orthogonal to the deltoids line of symmetry,
	\item
	a parallelogram.  
	\end{enumerate}
\end{enumerate}
Again Types 1--3 can be seen as the corresponding limits of Bricard octahedra only the Type 4 cases do not fall into this category.   
\end{enumerate}

\section{Bennett loops}\label{sec:bennett}

There is a huge body of literature on this topic, but we lean on \cite[Chapter 2]{dietmaier} due to its clarity. We consider a closed serial 4R chain  consisting of
 the systems $\Sigma_1,\ldots,\Sigma_4$, which are linked in a loop by rotary axes  $\go r_{1,4},\go r_{1,2}, \go r_{2,3}, \go r_{3,4}$, where $\go r_{i,j}$ denotes the connection between $\Sigma_i$ and $\Sigma_j$ with $i<j$. 
It is well-known that a 4R loop,
  which is neither planar ($\Leftrightarrow$ four R-joints are parallel) nor 
spherical ($\Leftrightarrow$ four R-joints are copunctal), can only be flexible iff it is a Bennett loop, 
which can be characterized 
by the following conditions on the Denavit-Hartenberg parameters:
\begin{equation}\label{eq:Bennett_conditions}
d_2=d_4, \quad d_1=d_3, \quad \alpha_2=\alpha_4, \quad \alpha_1=\alpha_3, \quad {d_1}{\sin{(\alpha_2)}}={d_2}{\sin{(\alpha_1)}}, 
\end{equation}
and zero offsets, where $d_i$ is the perpendicular distance between adjacent axes $\go r_{i-1,i}$ and $\go r_{i,i+1}$ 
and $\alpha_i$ is the twist angle enclosed by their directions. 
Later on we also need the point $\go F_{i,j}\in \go r_{i,j}$ 
in which the two common perpendiculars to the adjacent axes intersect $\go r_{i,j}$ (see Fig.\ \ref{fig1}). 

Note that Eq.\ (\ref{eq:Bennett_conditions}) holds under consideration of the following convention: 

\begin{convention}\label{rem:assumption}
According to \cite[page 9]{dietmaier} we can assume without loss of generality that (a) distances $d_1,\ldots, d_4>0$ and (b) $\alpha_1,\ldots, \alpha_4\in(0;\pi)$ hold. \hfill $\diamond$
\end{convention}

For the rotation angles $\theta_{i-1,i}$ about the axes $\go r_{i-1,i}$ the following relations hold:
\begin{equation}\label{eq:Bennett_relations}
\theta_{1,2}=-\theta_{3,4}, \quad \theta_{2,3}=-\theta_{1,4}, \quad K:=\tan\tfrac{\theta_{1,2}}{2}\tan\tfrac{\theta_{2,3}}{2}=\tfrac{\tan\tfrac{\alpha_1}{2}+\tan\tfrac{\alpha_2}{2}}{\tan\tfrac{\alpha_1}{2}-\tan\tfrac{\alpha_2}{2}}
\end{equation}
under Convention \ref{rem:assumption}. Note that $K$ is the transmission relation of adjacent rotation axes, which is also known as input-output equation.

\begin{remark}
Moreover, Dietmaier pointed out that one has to characterize a Bennett loop by the three 
values $\alpha_1$, $\alpha_2$ and either $d_1$ or $d_2$. The three values $d_1$, $d_2$ and $\alpha_i$ yield two possibilities for $\alpha_j$ 
with pairwise distinct $i,j\in\left\{1,2\right\}$.  \hfill $\diamond$
\end{remark}

\begin{figure}[t]
\begin{center}
\begin{overpic}
   [height=30mm]{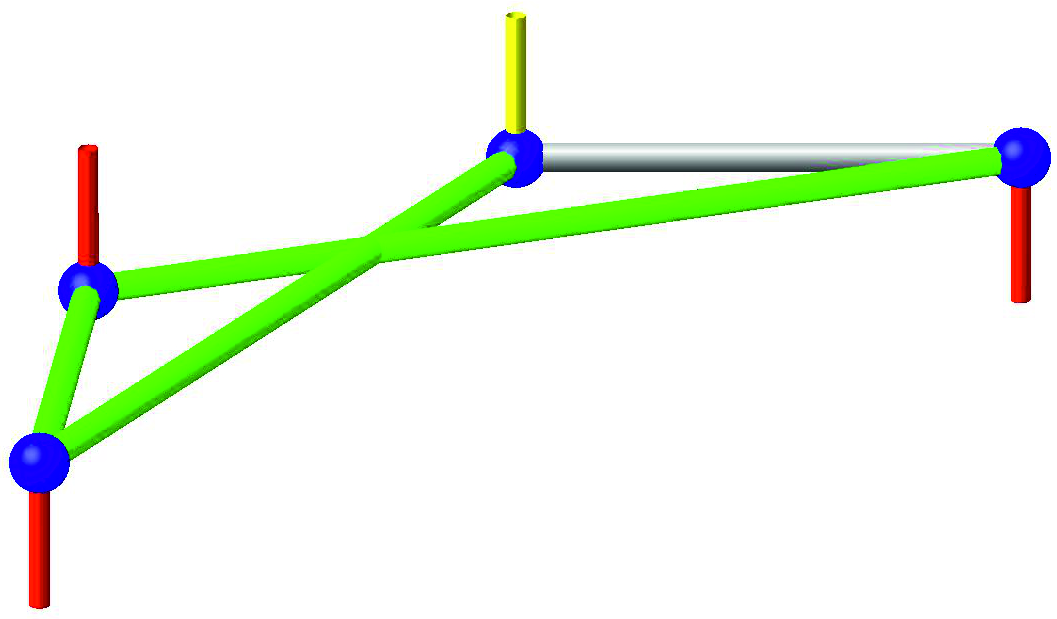}
\begin{scriptsize}
\put(-4,0){a)}
\put(38,45){$\go F_{1,4}$}
\end{scriptsize}     
  \end{overpic} 
\quad
\begin{overpic}
    [height=30mm]{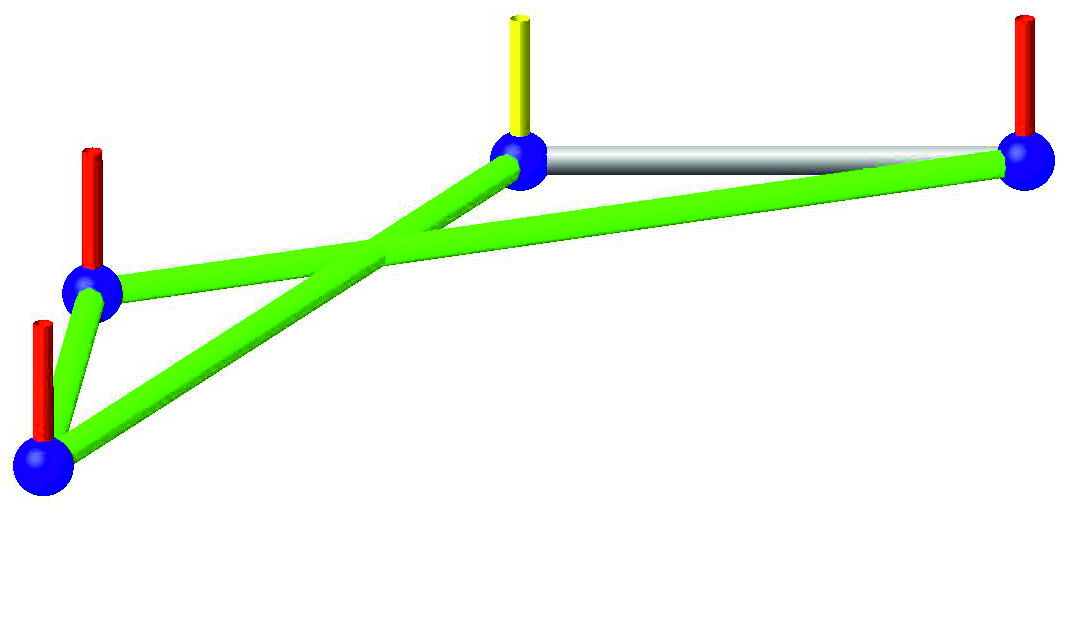}
\begin{scriptsize}
\put(-4,0){b)}
\put(38,45){$\go F_{1,4}$}
\end{scriptsize}     
  \end{overpic} 
\newline
\begin{overpic}
   [height=30mm]{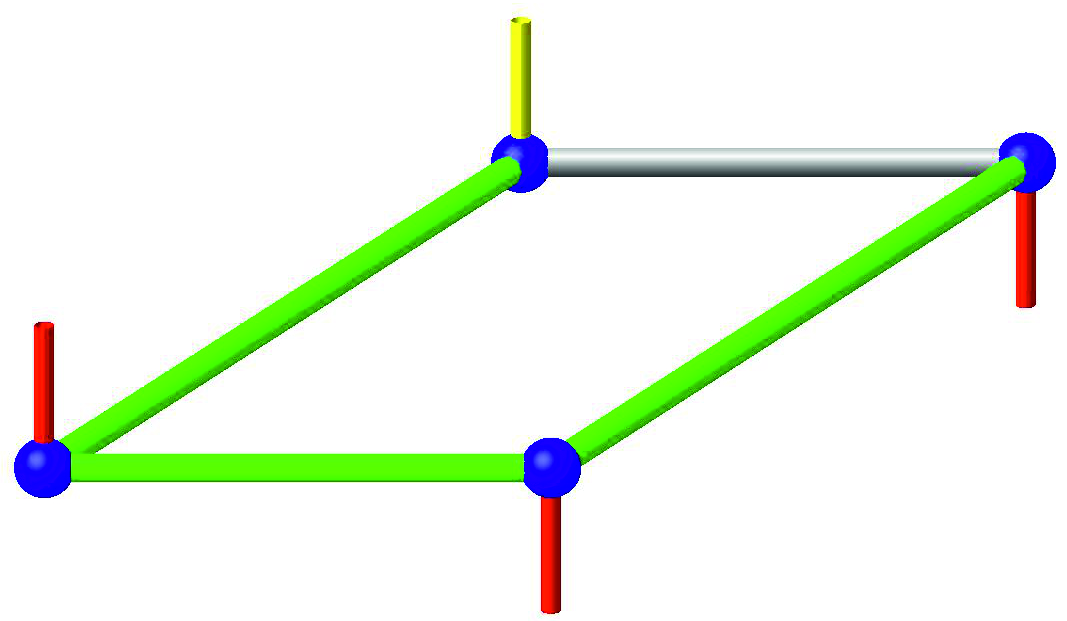}
\begin{scriptsize}
\put(-4,0){c)}
\put(38,45){$\go F_{1,4}$}
\end{scriptsize}     
  \end{overpic} 
\quad
\begin{overpic}
    [height=30mm]{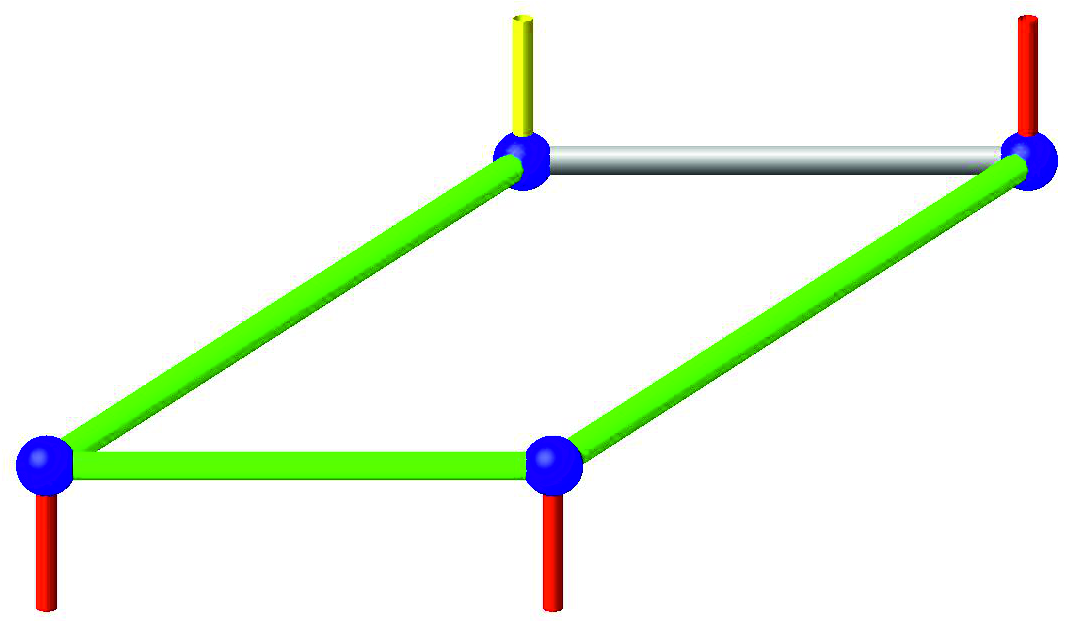}
\begin{scriptsize}
\put(-4,0){d)}
\put(38,45){$\go F_{1,4}$}
\end{scriptsize}     
  \end{overpic} 
  
\end{center}	
\caption{Illustration of the planar limits of Bennett loops: (a) Case (1a), (b) Case (2a), (c) Case (1b), (d) Case (2b). 
The illustrations are done for the parameters $d_1=0.5$, $d_2=1$ and $\tau=0.6$ using the DH-procedure discussed in Section \ref{sec:dhproc} (Remark \ref{rem:dh}), where the reference frame is chosen in such a way that its origin equals $\go F_{1,4}$, the $x$-axis (yellow) coincides with $\go r_{1,4}$ (incl.\ orientation) and $\go r_{1,2}$ intersects the positive $z$-axis (gray).
}
  \label{fig3}
\end{figure}
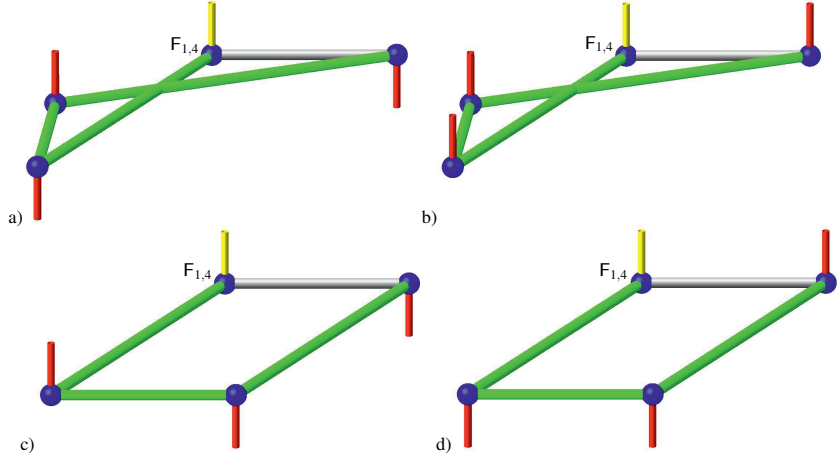

\subsection{Planar limits of Bennett loops}\label{sec:planar}

In this section we analyze the limits of the Bennett mechanism, when all axes become parallel; i.e.\ limit corresponds to a planar 4R-loop. 
More precisely, one can distinguish the following four cases as adjacent axes  $\go r_{i-1,i}$ and $\go r_{i,i+1}$ can be parallel ($\Leftrightarrow$ $\alpha_i=0$) or anti-parallel ($\Leftrightarrow$ $\alpha_i=\pi$). 
\begin{enumerate}
    \item $\alpha_1=\pi$
    \begin{enumerate}
        \item $\alpha_2=\pi$
        \item $\alpha_2=0$
    \end{enumerate}
    \item $\alpha_1=0$
    \begin{enumerate}
        \item $\alpha_2=0$
        \item $\alpha_2=\pi$
    \end{enumerate}
\end{enumerate}

For these four possible cases the corresponding limits of the 
transmission relation $K$ of 
Eq.\  (\ref{eq:Bennett_relations}) have to be computed. For doing this we first rewrite $K$ in dependence of the distances $d_1$ and $d_2$ by using either the identity
\begin{equation}\label{eq:ident1}
\tan\tfrac{\alpha_{i}}{2}=\tfrac{1-\cos(\alpha_i)}{\sin(\alpha_i)}
\end{equation}
or the following one:
\begin{equation}\label{eq:ident2}
\tan\tfrac{\alpha_{i}}{2}=\tfrac{\sin(\alpha_i)}{1+\cos(\alpha_i)}
\end{equation}
Thus two cases have to be distinguished:
\begin{enumerate}
    \item 
    Let's make use of Eq.\ (\ref{eq:ident1}) which turns $K$ into:
    \begin{equation*}
    K=\tfrac{\tfrac{1-\cos(\alpha_1)}{\sin(\alpha_1)}+\tfrac{1-\cos(\alpha_2)}{\sin(\alpha_2)}}{\tfrac{1-\cos(\alpha_1)}{\sin(\alpha_1)}-\tfrac{1-\cos(\alpha_2)}{\sin(\alpha_2)}}.
    \end{equation*}
    Now we can set $\sin(\alpha_i)=\tfrac{d_i}{k}$ with a constant $k\in\RR^+\setminus\left\{0\right\}$, which can be done due to the last relation of Eq.\ (\ref{eq:Bennett_conditions}), which yields 
    \begin{equation}\label{Kd:1}
    K=\tfrac{\tfrac{1-\cos(\alpha_1)}{d_1}+\tfrac{1-\cos(\alpha_2)}{d_2}}{\tfrac{1-\cos(\alpha_1)}{d_1}-\tfrac{1-\cos(\alpha_2)}{d_2}}
    \end{equation}
    as $k$ cancels out. By considering the limit  $\alpha_1=\pi$ this equation simplifies to 
    \begin{equation*}
    K=\tfrac{d_1+2d_2-d_1\cos(\alpha_2)}{-d_1+2d_2+d_1\cos(\alpha_2)}
    \end{equation*}
    Now we remain with two possible limits for $\alpha_2$:
        \begin{enumerate}
        \item $\alpha_2=\pi$: This implies $K=\tfrac{d_1+d_2}{d_2-d_1}$.
        \item $\alpha_2=0$: Then we get $K=1$.
        \end{enumerate}
    \item 
    By making use of Eq.\ (\ref{eq:ident2}) we get the following expression from the last equation of Eq.\ (\ref{eq:Bennett_relations}):
    \begin{equation*}
    K=\tfrac{\tfrac{\sin(\alpha_1)}{1+\cos(\alpha_1)}+\tfrac{\sin(\alpha_2)}{1+\cos(\alpha_2)}}{\tfrac{\sin(\alpha_1)}{1+\cos(\alpha_1)}-\tfrac{\sin(\alpha_2)}{1+\cos(\alpha_2)}}
    \end{equation*}
    Now we can again set  $\sin(\alpha_i)=\tfrac{d_i}{k}$ with a constant $k\in\RR^+\setminus\left\{0\right\}$, which yields 
    \begin{equation}\label{Kd:2}
    K=\tfrac{\tfrac{d_1}{1+\cos(\alpha_1)}+\tfrac{d_2}{1+\cos(\alpha_2)}}{\tfrac{d_1}{1+\cos(\alpha_1)}-\tfrac{d_2}{1+\cos(\alpha_2)}}
    \end{equation}
   as again $k$ cancels out.  For $\alpha_1=0$ this expression  simplifies to 
    \begin{equation*}
    K=\tfrac{d_1+2d_2+d_1\cos(\alpha_2)}{d_1-2d_2+d_1\cos(\alpha_2)}
    \end{equation*}
    Now we remain with two possible limits for $\alpha_2$:
        \begin{enumerate}
        \item $\alpha_2=0$: Then we get $K=\tfrac{d_1+d_2}{d_1-d_2}$.
        \item $\alpha_2=\pi$: This implies $K=-1$.
        \end{enumerate}
\end{enumerate}

\begin{remark}
Note that putting $\alpha_1=\alpha_2=0$ into Eq.\ (\ref{Kd:1}) or $\alpha_1=\alpha_2=\pi$ into
 Eq.\ (\ref{Kd:2})
would yield mathematically undefined expressions. \hfill $\diamond$
\end{remark}

It can easily be seen that cases (1a) and (2a) correspond to the anti-parallelogram and the cases (1b) and (2b) to the parallelogram, respectively. These cases are also illustrated in Fig.\ \ref{fig3}.

\begin{remark}
The cases (1x) and (2x) are related by reversing the orientation of two opposite rotation axes for x$\in\left\{\text{a,b}\right\}$. 
Note that another approach for deducing that the anti-parallelogram and parallelogram, respectively, are the planar limit cases of the Bennett mechanism is given in \cite{herve}.  
\hfill $\diamond$ 
\end{remark}

\begin{figure}[t]
\begin{center}
\begin{overpic}
   [height=45mm]{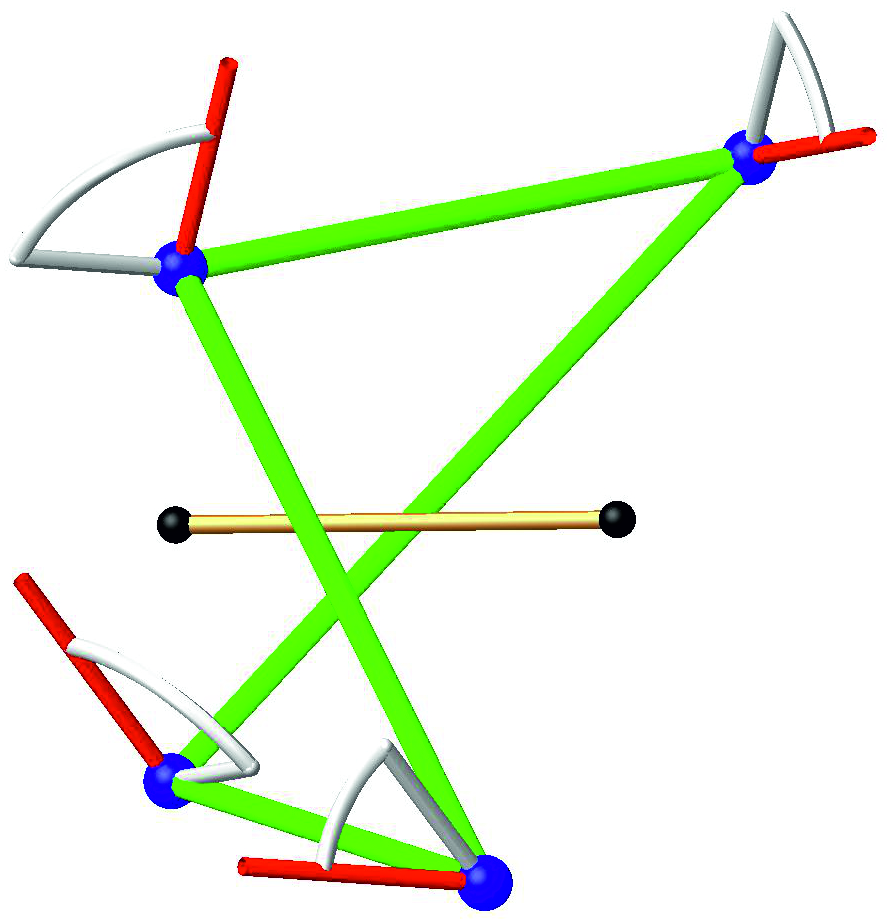}
\begin{scriptsize}
\put(0,0){a)}
\put(11,8){$\go F_{2,3}$}
\put(0,25){$\go r_{2,3}$}
\put(15,28){$\alpha_2$}
\put(25,0){$\go r_{3,4}$}
\put(57,2){$\go F_{3,4}$}
\put(32,14){$\alpha_1$}
\put(88,79){$\go r_{1,2}$}
\put(70,87){$\go F_{1,2}$}
\put(90,92){$\alpha_1$}
\put(9,65){$\go F_{1,4}$}
\put(27,90){$\go r_{1,4}$}
\put(6,84){$\alpha_2$}
\end{scriptsize}     
  \end{overpic} 
\,
\begin{overpic}
    [height=30mm]{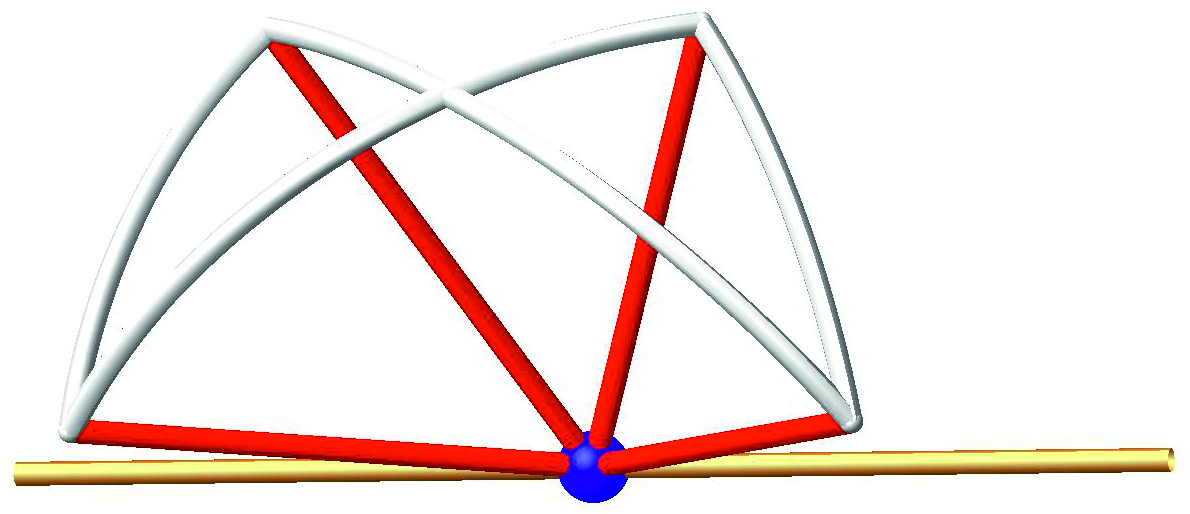}
\begin{scriptsize}
\put(-5,0){b)}
\put(68,27){$\alpha_1$}
\put(6,27){$\alpha_1$}
\put(20,21){$\alpha_2$}
\put(42,27){$\alpha_2$}
\put(32,7.5){$\Vkt r_{3,4}$}
\put(58,8.5){$\Vkt r_{1,2}$}
\put(34,15){$\Vkt r_{2,3}$}
\put(54.5,14){$\Vkt r_{1,4}$}
\end{scriptsize}     
  \end{overpic} 
\end{center}	
\caption{(a) Same Bennett mechanism as illustrated in Fig.\ \ref{fig1} but viewed from another perspective. A Bennett mechanism is line-symmetric, where the symmetry line is indicated in yellow. The black points on it are the midpoints of $\go F_{1,4}$, $\go F_{2,3}$ and  $\go F_{1,2}$, $\go F_{3,4}$, respectively. (b) Corresponding spherical indicatrix, which is a spherical anti-parallelogram. Note that the direction of the symmetry line is orthogonal to the symmetry plane of the spherical anti-parallelogram.}
  \label{fig2}
\end{figure}
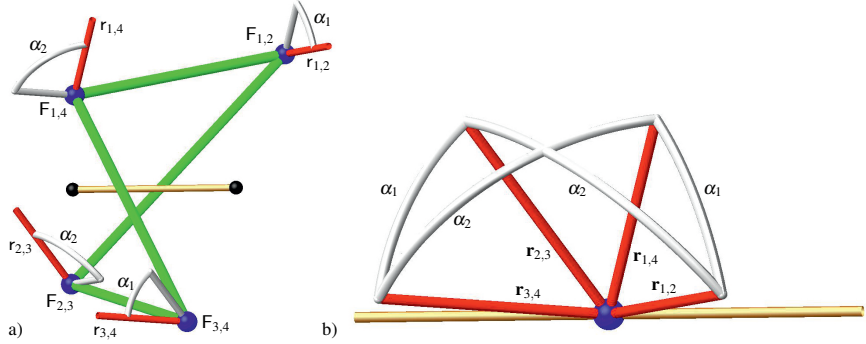

\subsection{Spherical limits of Bennett loops}\label{sec:spherical}

We remain with analyzing the limit of the Bennett mechanism, when all the distances vanish ($d_1=d_2=0$).  
This is equivalent to considering the spherical indicatrix (see \cite{hunt}) of the Bennett mechanism, which is a spherical 4R-loop. The transmission relation $K$ of Eq.\ (\ref{eq:Bennett_relations})
is invariant under this limiting process as it does not depend on $d_1$ and $d_2$, respectively. It was already pointed out by Baker \cite{baker_bennett} that the spherical indicatrix 
is a spherical anti-parallelogram  with respect to Convention \ref{rem:assumption}, which is illustrated in Fig.\ \ref{fig2}. 
Note that for such a spherical 4-bar mechanism we can always assume without loss of generality that the spherical lengths are less than $\pi$.

Assuming that a Bennett mechanism is given, one can also assign in an arbitrary way the orientations to the axes. A different orientation of an axis compared to Convention \ref{rem:assumption} means that in the spherical indicatrix the corresponding point on the unit-sphere (Fig.\ \ref{fig2}) has to be replaced by its antipodal point. All cases, which were extensively discussed by the author in \cite{kilian}, can be subdivided into two classes; namely the Voss class and anti-Voss one. Their definitions read as follows (cf.\ \cite[Definition 1]{kilian}), under the assumption that spherical lengths of opposite sides of the spherical indicatrix are denoted by
$\alpha$, $\alpha^*$ and $\beta$, $\beta^*$, respectively:
\begin{definition}\label{def:voss}
The spherical indicatrix of a 4R-loop is called 
\begin{enumerate}[$\bullet$]
    \item 
    V-hedral for  $\alpha^*=\alpha$
    and  $\beta^*=\beta$.
    \item 
    anti-V-hedral for  $\alpha^*=\pi-\alpha$ 
    and  $\beta^*=\pi-\beta$. 
\end{enumerate}
\end{definition}

 However, it is possible to unify these two classes using spherical kinematics, where it is known as the isogonal case \cite{stachel2}. We only want to use this notation in the remainder of the paper as it avoids a cumbersome distinction of cases.

%%%%%%%%%%%%%%%%%%%%%%%%%%%%%%%%%%%%%%%%%%%

\subsection{Denavit-Hartenberg procedure}\label{sec:dhproc}

For the computation we attach the fixed reference frame in the following way to $\Sigma_1$: The $x$-axis is aligned with the axis $\go r_{1,4}$ and the origin is placed in the point $\go F_{1,4}$. Moreover, the positive $z$-axis intersects the axis $\go r_{1,2}$. 

In the following we only need two types of transformation matrices:
\begin{equation}\label{eq:matrices}
\Vkt R(\theta_{i,j})=\begin{pmatrix}
1 & 0 & 0 & 0 \\
0 & 1 & 0 & 0 \\
0 & 0 & \cos(\theta_{i,j}) & \sin(\theta_{i,j}) \\
0 & 0 & -\sin(\theta_{i,j}) & \cos(\theta_{i,j}) 
\end{pmatrix}, \,\,
\Vkt T(\alpha_i)=\begin{pmatrix}
1 & 0 & 0 & 0 \\
0 & \cos(\alpha_i) & -\sin(\alpha_i) & 0 \\
0 & \sin(\alpha_i) & \cos(\alpha_i) & 0 \\
d_i & 0 & 0 & 1 
\end{pmatrix},
\end{equation}
where we set $d_i=k\sin(\alpha_i)$ with a constant $k\in\RR^+\setminus\left\{0\right\}$, which can be done due to the last relation of Eq.\ (\ref{eq:Bennett_conditions}).
As we want to end up with algebraic expressions we make the half-angle substitutions 
\begin{equation}\label{eq:halfanglesubsti2}
a_i=\tan\tfrac{\alpha_i}{2}, \quad 
t_{i,j}=\tan\tfrac{\theta_{i,j}}{2},
\end{equation}
which we insert into Eq.\ (\ref{eq:matrices}) yielding $\Vkt R(t_{i,j})$ and $\Vkt T(a_i)$, respectively. 

\begin{remark}\label{rem:Bennett_relations}
Note that the half-angle substitutions of Eq.\ (\ref{eq:halfanglesubsti2}) simplify the last condition of Eq.\ (\ref{eq:Bennett_relations}) to $t_{1,2}t_{2,3}=\tfrac{a_1+a_2}{a_1-a_2}$. \hfill $\diamond$
\end{remark}

In the following we give the matrices $\Vkt M_{i,j}$ which can be used to compute the axis $\go r_{i,j}$ with respect to the fixed frame:
\begin{equation} \label{eq:upper}
\begin{split}
\Vkt M_{1,2}&:=\Vkt T(a_1),\\
\Vkt M_{2,3}&:=\Vkt T(a_1)\Vkt R(t_{1,2})\Vkt T(a_2), \\
\Vkt M_{3,4}&:=\Vkt T(a_1)\Vkt R(t_{1,2})\Vkt T(a_2)\Vkt R(t_{2,3})\Vkt T(a_1).
\end{split}
\end{equation}

It can easily be checked (e.g.\ with Maple) that the loop closure condition
\begin{equation}\label{eq:loop}
\Vkt T(a_1)\Vkt R(t_{1,2})\Vkt T(a_2)\Vkt R(t_{2,3})\Vkt T(a_1)\Vkt R(-t_{1,2})\Vkt T(a_2)\Vkt R(-t_{2,3})=\Vkt I_4    
\end{equation}
holds true under consideration of the transmission relation $K$, where  $\Vkt I_4$ denotes the ($4 \times 4$)-identity matrix. 

\begin{remark}\label{rem:dh}
Clearly, the correctness of the loop closure condition (\ref{eq:loop}) can also be checked for the spherical limit of the Bennett mechanism discussed in Section \ref{sec:spherical}, which can easily be obtained by setting $k=0$ ($\Rightarrow d_1=d_2=0$). 

For the planar limits of Section \ref{sec:planar}, we do not make the substitutions
 $d_i=k\sin(\alpha_i)$  but keep the unknowns $d_i$ within the matrices $\Vkt T(\alpha_i)$ of Eq.\ (\ref{eq:matrices}). By using the corresponding 
 transmission ratio $K$ also the resulting loop closure conditions can be validated. 
\hfill $\diamond$
\end{remark}

For the case of Bennett mechanisms, we set $t_{1,2}=\tfrac{a_1+a_2}{(a_1-a_2)t_{2,3}}$. 
Finally, we rename the remaining parameter $t_{2,3}$ by $\tau$ for the remainder of the paper. 
In this way we get from Eq.\ (\ref{eq:upper}) 
the transformation matrices $\Vkt M_{i,j}(\tau)$. Note that $\Vkt M_{i,j}(\tau)$ depends only on the following three design parameters
$k,a_1,a_2$ beside the motion parameter $\tau$.

Now we can compute the unit direction vector $\Vkt r_{i,j}$ of the axis $\go r_{i,j}$ 
and the position vector $\Vkt F_{i,j}$ of the distinguished point $\go F_{i,j}\in \go r_{i,j}$ by:
\begin{equation}\label{eq:Fandr}
\begin{pmatrix}
1 \\
\Vkt F_{i,j}
\end{pmatrix} =\Vkt M_{i,j}(\tau)
\begin{small}
\begin{pmatrix}
1 \\ 0 \\ 0 \\ 0 
\end{pmatrix}
\end{small}
,
\quad
\begin{pmatrix}
0 \\
\Vkt r_{i,j}
\end{pmatrix} =\Vkt M_{i,j}(\tau)
\begin{small}
\begin{pmatrix}
0 \\ 1 \\ 0 \\ 0 
\end{pmatrix}
\end{small}
.
\end{equation}

%%%%%%%%%%%%%%%%%%%%%%%%%%%%%%%%%%%%%%%%%%%%%%%%%%%%%%%%%%%%%%%%%%%%%%%%%%%%%%%%%%%%%%%%%%%%%%%%%%%%%%%%%%%%%%%%%%%%%%%%%%%%%

\subsection{Remarkable properties and special cases}\label{sec:prop}

It is well-known (e.g.\ see \cite{hon}) that a Bennett mechanism is line-symmetric (see Fig.\ \ref{fig2}) if neglecting the orientation of the rotation axes\footnote{If we take the orientation of the axes with respect to Convention \ref{rem:assumption} into account then it is not line-symmetric as already pointed out by Baker \cite{baker_bennett}.}. 
This symmetry property implies that the relative motion of opposite links is
line-symmetric, which was intensively studied by Krames \cite{krames}. During 
the motion, the symmetry-lines swap a one-sheeted hyperboloid (see also \cite{hamann} and \cite[pages 326-328]{bottema}). Moreover, note that the 
four axes of the Bennett mechanism are always located in a regulus of lines \cite{hon}.

\begin{remark}
We could change Convention \ref{rem:assumption} by assuming 
(a) distances $d_1=d_3>0$, $d_2=d_4>0$ and (b) $\alpha_1=\alpha_3\in(0;\pi)$, $\alpha_1=\alpha_3\in(-\pi;0)$. This corresponds to the change of the orientation of the axes $\go r_{2,3}$ and $\go r_{3,4}$ and also of the corresponding rotation angles $\theta_{2,3}$ and $\theta_{3,4}$. Concerning this convention Eq.\ (\ref{eq:Bennett_relations}) would 
read:
\begin{equation}\label{eq:Bennett_relations_sym}
\theta_{1,2}=\theta_{3,4}, \quad \theta_{2,3}=\theta_{1,4}, \quad K:=t_{1,2}t_{2,3}=\tfrac{1-a_1a_2}{1+a_1a_2}
\end{equation}
as $a_2$ in Eq.\ (\ref{eq:Bennett_relations})  has to be replaced by $-\tfrac{1}{a_2}$. Moreover, the last relation of 
 Eq.\ (\ref{eq:Bennett_conditions})
would also get a sign change; i.e.\ 
${d_1}{\sin{(\alpha_2)}}=-{d_2}{\sin{(\alpha_1)}}$. 
Concerning this convention the Bennett mechanism would be line-symmetric including also the orientation of the axes, which was also pointed out in \cite{song}.  
The corresponding spherical indicatrix would be an spherical parallelogram. According to Baker \cite{baker6R} these two different conventions, have been the source of much confusion in the literature.
\hfill $\diamond$
\end{remark}

There is also a subset of Bennett mechanisms having an additional plane-symmetry. For this subset the line-symmetry results from two plane-symmetries where the corresponding planes intersect orthogonally (along the axis of the line-symmetry). 
It can easily be checked that this additional plane-symmetry is equivalent with the condition that opposite axes of the Bennett mechanism intersect, which is again equivalent with the condition\footnote{This can easily be checked by computing the Plücker coordinates of the axes in dependence of $\tau$ based on Section \ref{sec:dhproc} and applying the intersection condition \cite{pottmann}.}
$a_1a_2=1$. Note that this is also equivalent with the condition that the four common normals have the same lengths ($\Leftrightarrow$ $d_1=d_2$) assumed $a_1\neq a_2$ (see also \cite{hon}). 
In this special case the regulus spanned by the four axes of the Bennett mechanism, splits up into two pencils of lines, whose vertices are located on the intersection line of the corresponding two carrier planes\footnote{The conic in the Klein image splits up into two lines \cite{pottmann}.}.

%%%%%%%%%%%%%%%%%%%%%%%%%%%%%%%%%%%%%%%%%%%%%%%%%%%%%%%%%%%%%%%%%%%%%%%%%%%%%%%%%%%%%%%%%%%%%%%%%%%%%%%%%%%%%%%%%%%%%%%%%%%%%
%%%%%%%%%%%%%%%%%%%%%%%%%%%%%%%%%%%%%%%%%%%%%%%%%%%%%%%%%%%%%%%%%%%%%%%%%%%%%%%%%%%%%%%%%%%%%%%%%%%%%%%%%%%%%%%%%%%%%%%%%%%%%

\section{Flexible arrangement of two Bennett tubes}\label{sec:pairedB}

In this section we present first classes of flexible arrangements of two Bennett tubes, which are denoted by $\mathcal{B}$ and $\dach{\mathcal{B}}$. In Subsection \ref{subsec:sym} we use for the construction a symmetry relation between  $\mathcal{B}$ and $\dach{\mathcal{B}}$. In Section \ref{subsec:nonsym} we present an algebraic approach for the determination of general flexible bi-Bennetts, but we are only able to solve the resulting set of necessary and sufficient conditions for a special family. 

\subsection{Symmetric arrangements}\label{subsec:sym}

Let us start by parametrizing the points $\go P_{i,j}$ on the axis $\go r_{i,j}$ of the Bennett $\mathcal{B}$; i.e.\
\begin{equation}\label{eq:pij}
\Vkt P_{i,j}=\Vkt F_{i,j}+\mu_{i,j}\Vkt r_{i,j}    
\end{equation}
with $\Vkt F_{i,j}$ and $\Vkt r_{i,j}$ of Eq.\ (\ref{eq:Fandr}). 

The first idea that comes into mind is to choose the $\mu_{i,j}$'s in way that the four points $\go P_{1,4},\go P_{1,2}, \go P_{2,3}, \go P_{3,4}$
are coplanar for all $\tau\in\RR$, as then one can reflect $\mathcal{B}$ on the corresponding plane to obtain $\dach{\mathcal{B}}$. Unfortunately, such a plane-symmetric flexible bi-Bennett  cannot exist, which can easily be checked (see Appendix).

Therefore, we only remain with the line-symmetry. According to Cayley  \cite{cayley} the skew quadrilateral $\go P_{1,4},\go P_{1,2}, \go P_{2,3}, \go P_{3,4}$ is line-symmetric if and only if it is a skew isogram; i.e.\ opposite sides have the same length which is expressed by the following two equations:
\begin{equation}\label{eq:linsym}
    \|\Vkt P_{1,4}-\Vkt P_{1,2} \|= \|\Vkt P_{2,3}-\Vkt P_{3,4} \| \qquad
     \|\Vkt P_{1,4}-\Vkt P_{3,4} \|= \|\Vkt P_{2,3}-\Vkt P_{1,2} \| 
\end{equation}
As these two equations are independent of $\tau$, we only have two conditions on the seven unknowns $k,a_1,a_2,\mu_{1,4},\mu_{1,2},\mu_{2,3},\mu_{3,4}$. We can cancel the factor of similarity by setting $k=1$. Then  we solve Eq.\ (\ref{eq:linsym}) for $a_1$ and $a_2$ which yields
\begin{equation}\label{eq:as}
\begin{split}
   a_1&= \sqrt{-\tfrac{(\mu_{1,4} - \mu_{1,2} + \mu_{2,3} - \mu_{3,4})(\mu_{1,4} - \mu_{1,2} - \mu_{2,3} + \mu_{3,4})}
    {(\mu_{1,4} + \mu_{1,2} + \mu_{2,3} + \mu_{3,4})(\mu_{1,4} + \mu_{1,2} - \mu_{2,3} - \mu_{3,4})}} \\
    a_2&= \sqrt{-\tfrac{(\mu_{1,4} + \mu_{1,2} - \mu_{2,3} - \mu_{3,4})(\mu_{1,4} - \mu_{1,2} + \mu_{2,3} - \mu_{3,4})}
    {(\mu_{1,4} + \mu_{1,2} + \mu_{2,3} + \mu_{3,4})(\mu_{1,4} - \mu_{1,2} - \mu_{2,3} + \mu_{3,4})}}
    \end{split}
\end{equation}
By a half-turn of $\mathcal{B}$ about the line of symmetry of  $\go P_{1,4},\go P_{1,2}, \go P_{2,3}, \go P_{3,4}$ we obtain $\dach{\mathcal{B}}$. This yields a 4-dimensional solution set.

\begin{figure}[t]
\begin{center}
\begin{overpic}
   [height=70mm]{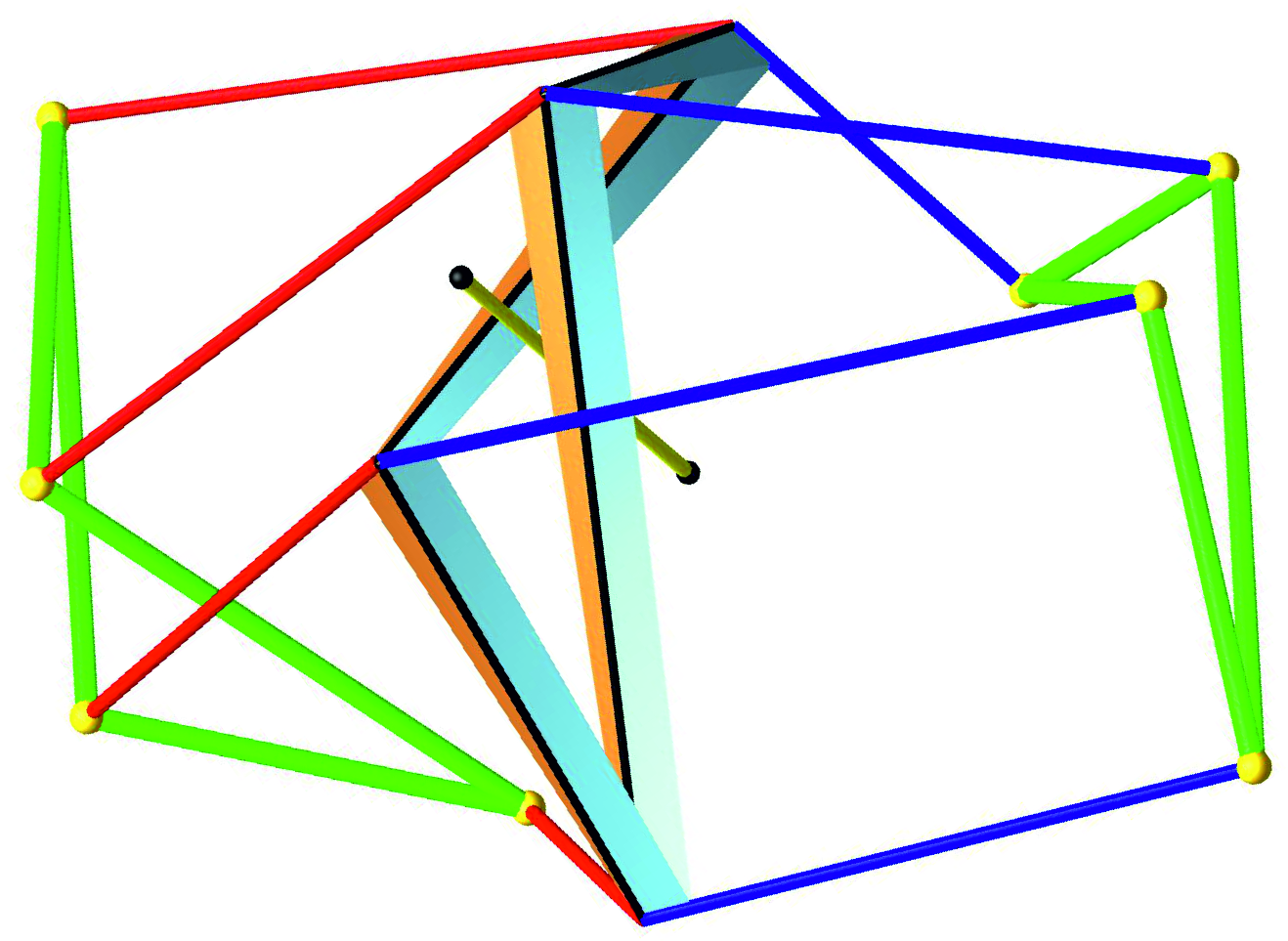}
\begin{scriptsize}
\put(74,5){$\go r_{1,2}$}
\put(94,9.5){$\go F_{1,2}$}
\put(50,-1){$\go P_{1,2}$}
\put(65,42){$\go r_{2,3}$}
\put(89,52){$\go F_{2,3}$}
\put(24,38){$\go P_{2,3}$}
\put(67,56){$\go r_{3,4}$}
\put(73,50){$\go F_{3,4}$}
\put(57,72){$\go P_{3,4}$}
\put(76,64){$\go r_{1,4}$}
\put(94,62){$\go F_{1,4}$}
\put(35.5,65.2){$\go P_{1,4}$}
\put(24,69){$\dach{\go r}_{1,2}$}
\put(3,66.5){$\dach{\go F}_{1,2}$}
\put(20,53){$\dach{\go r}_{2,3}$}
\put(0,31.5){$\dach{\go F}_{2,3}$}
\put(18,32){$\dach{\go r}_{1,4}$}
\put(3,14){$\dach{\go F}_{1,4}$}
\put(42.5,3.5){$\dach{\go r}_{3,4}$}
\put(37,7){$\dach{\go F}_{3,4}$}
\end{scriptsize}     
  \end{overpic} 
\end{center}	
\caption{Line-symmetric flexible bi-Bennett of family (A) of Theorem \ref{thm:linsym}: The blue Bennett $\mathcal{B}$ and the red Bennett $\dach{\mathcal{B}}$ are related by the line-symmetry of the skew isogram $\go P_{1,4},\go P_{1,2}, \go P_{2,3}, \go P_{3,4}$, where the symmetry line is displayed in yellow. The black points on it are the midpoints of $\go P_{1,4}$, $\go P_{2,3}$ and  $\go P_{1,2}$, $\go P_{3,4}$, respectively.
Along this isogram the two Bennett tubes are indicated by a self-collision-free slim ribbon of orange and light-blue HP-patches. 
Note that the vertices $\go P_{1,2}$ and $\go P_{3,4}$ are anti-Voss and $\go P_{1,4}$ and $\go P_{2,3}$ are Voss (cf.\ Definition \ref{def:voss}). 
This illustration corresponds to the following parameters: $k=1$, $\mu_{2,3}=1$, $\mu_{3,4}=\tfrac{1}{2}$, $\mu_{1,4}=\tfrac{37}{40}$, $\mu_{1,2}=\tfrac{7}{8}$, $a_1=\tfrac{1}{2}$, $a_2=\tfrac{1}{3}$ and $\tau=\tfrac{9}{10}$. 
}
  \label{fig4}
\end{figure}
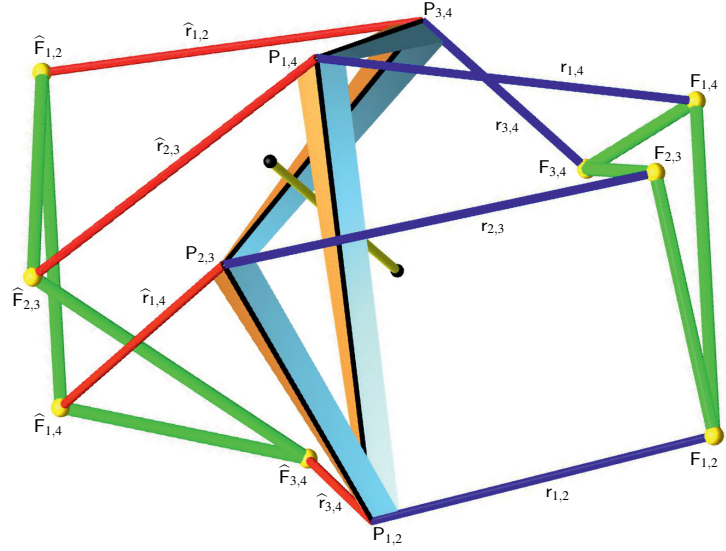

It can easily be checked that there are only two special cases where the equations of Eq.\ (\ref{eq:linsym}) cannot be solved for $a_1$ and $a_2$ as given in Eq.\ (\ref{eq:as}), which are as follows:
\begin{align}
  &\mu_{1,4} = \mu_{2,3},&\quad  &\mu_{1,2} = \mu_{3,4}  \label{eq1} \\
  &\mu_{1,4} = -\mu_{2,3},&\quad  &\mu_{1,2} = -\mu_{3,4} \label{eq2} 
\end{align}
In case of Eq.\ (\ref{eq2}) the line of symmetry of $\go P_{1,4},\go P_{1,2}, \go P_{2,3}, \go P_{3,4}$ coincides with the one of $\go F_{1,4},\go F_{1,2}, \go F_{2,3}, \go F_{3,4}$. Therefore, $\mathcal{B}$ equals $\dach{\mathcal{B}}$, which is a trivial solution to our problem. 
Therefore, we remain with the case of Eq.\ (\ref{eq1}), which also gives a 4-dimensional solution set. This already proves the following result:
\begin{theorem}\label{thm:linsym}
    Besides the factor of similarity there exists two 4-dimensional families of 
     flexible line-symmetric arrangements of two Bennett tubes. Family (A) is determined by Eq.\ (\ref{eq:as}) and family (B) by  Eq.\ (\ref{eq1}). 
\end{theorem}
Examples of flexible line-symmetric bi-Bennetts for each family are 
illustrated in Figs.\ \ref{fig4} and \ref{fig5}.
It makes sense to distinguish these two families, as they possess different geometric properties, which are pointed out in the following theorem:

\begin{figure}[t]
\begin{center}
\begin{overpic}
   [height=70mm]{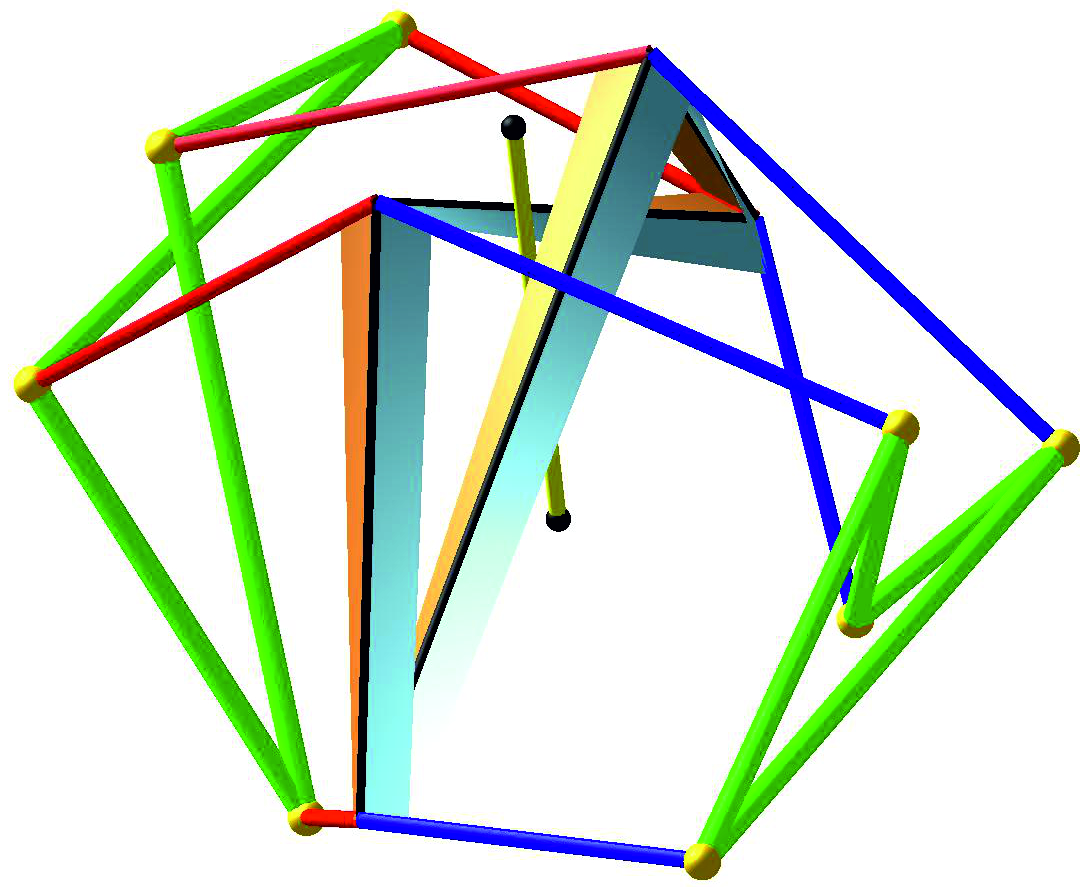}
\begin{scriptsize}
\put(67,0){$\go F_{1,2}$}
\put(33,3){$\go P_{1,2}$}
\put(22,4){$\dach{\go F}_{3,4}$}
\put(85,42){$\go F_{2,3}$}
\put(32,64.5){$\go P_{2,3}$}
\put(57,78.5){$\go P_{1,4}$}
\put(75,21){$\go F_{3,4}$}
\put(70.5,59){$\go P_{3,4}$}
\put(96,43){$\go F_{1,4}$}
\put(39,79){$\dach{\go F}_{1,2}$}
\put(9,66){$\dach{\go F}_{2,3}$}
\put(-2,43){$\dach{\go F}_{1,4}$}
\end{scriptsize}     
  \end{overpic} 
\end{center}	
\caption{Line-symmetric flexible bi-Bennett of family (B) of Theorem \ref{thm:linsym}. We only labeled the points (but not the rotation axes) to avoid an overloading of the figure. This illustration corresponds to the following parameters: $k=1$, $\mu_{2,3}=\tfrac{2}{3}$, $\mu_{3,4}=\tfrac{1}{2}$, $a_1=\tfrac{1}{2}$, $a_2=\tfrac{1}{3}$ and $\tau=\tfrac{9}{10}$.}
  \label{fig5}
\end{figure}

\begin{theorem}\label{thm:property}
The spherical indicatrix of each spherical 4R-loop with center  $\go P_{i,j}$ of family (A) is isogonal and of family (B) deltoidal\footnote{If one pair of adjacent sides has the same spherical length and the remaining two sides have also the same spherical length; i.e.\ $\alpha=\beta$ and $\alpha^*=\beta^*$ in the notation of Definition \ref{def:voss}.}  (symmetry plane is spanned by  $\go P_{i-1,j-1}$ and $\go P_{i+1,j+1}$). 
\end{theorem}

\begin{proof}
We do the proof exemplarily for the vertex\footnote{For all other vertices the proof is the same, there is just a shift of indices.} $P_{2,3}$. 
Let us start with proving the property of family (A). To do so, we  express the isogonality by a suitable set of equations and check if it is fulfilled for the values of $a_1$ and $a_2$ given in Eq.\ (\ref{eq:as}). 

We consider the pair of opposite angles $\sphericalangle\go P_{1,2}\go P_{2,3}\go F_{2,3}$ and
$\sphericalangle\go P_{3,4}\go P_{2,3}\dach{\go F}_{1,4}$, where the latter equals   $\sphericalangle\go P_{1,2}\go P_{1,4}\go F_{1,4}$ due to the line-symmetry (see Fig.\ \ref{fig4}). Therefore we can set up the following equation:
\begin{equation}\label{eq:iso1}
    \left\langle
    \frac{\Vkt P_{1,2}-\Vkt P_{2,3}}{\| \Vkt P_{1,2}-\Vkt P_{2,3}\|},\Vkt r_{2,3}
    \right\rangle^2
    -
    \left\langle
    \frac{\Vkt P_{1,2}-\Vkt P_{1,4}}{\| \Vkt P_{1,2}-\Vkt P_{1,4}\|},\Vkt r_{1,4}
    \right\rangle^2=0.
\end{equation}
Now we consider the other pair of opposite angles $\sphericalangle\go P_{3,4}\go P_{2,3}\go F_{2,3}$ and
$\sphericalangle\go P_{1,2}\go P_{2,3}\dach{\go F}_{1,4}$, where the latter equals   $\sphericalangle\go P_{3,4}\go P_{1,4}\go F_{1,4}$ due to the line-symmetry (see Fig.\ \ref{fig4}). Therefore we can set up the following equation:
\begin{equation}\label{eq:iso2}
    \left\langle
    \frac{\Vkt P_{3,4}-\Vkt P_{2,3}}{\| \Vkt P_{3,4}-\Vkt P_{2,3}\|},\Vkt r_{2,3}
    \right\rangle^2
    -
    \left\langle
    \frac{\Vkt P_{3,4}-\Vkt P_{1,4}}{\| \Vkt P_{3,4}-\Vkt P_{1,4}\|},\Vkt r_{1,4}
    \right\rangle^2=0.
\end{equation}
Due to the squaring the Eqs.\ (\ref{eq:iso1}) and (\ref{eq:iso2}) are algebraic and they are fulfilled if either the angles are equal or supplementary. But it is not allowed that one pair of opposite angles is equal and the other one supplementary. To cancel out this option we have to set up a further equation:
\begin{equation}\label{eq:iso3}
\begin{split}
&\left\langle
    \frac{\Vkt P_{1,2}-\Vkt P_{2,3}}{\| \Vkt P_{1,2}-\Vkt P_{2,3}\|},\Vkt r_{2,3}
    \right\rangle
    \left\langle
    \frac{\Vkt P_{3,4}-\Vkt P_{1,4}}{\| \Vkt P_{3,4}-\Vkt P_{1,4}\|},\Vkt r_{1,4}
    \right\rangle
    - \\
    &\left\langle
    \frac{\Vkt P_{1,2}-\Vkt P_{1,4}}{\| \Vkt P_{1,2}-\Vkt P_{1,4}\|},\Vkt r_{1,4}
    \right\rangle
    \left\langle
    \frac{\Vkt P_{3,4}-\Vkt P_{1,4}}{\| \Vkt P_{3,4}-\Vkt P_{1,4}\|},\Vkt r_{1,4}
    \right\rangle=0.
\end{split}
\end{equation}
Note that using the equalities of Eq.\ (\ref{eq:linsym}) renders this condition also algebraic. Now it can easily be checked by using e.g.\ Maple that the Eqs.\ (\ref{eq:iso1}--\ref{eq:iso3}) are fulfilled for the values of $a_1$ and $a_2$ given in Eq.\ (\ref{eq:as}). 
This finishes the first part of the proof. 

Now we prove the property of family (B).  
We consider the pair of adjacent angles $\sphericalangle\go P_{1,2}\go P_{2,3}\go F_{2,3}$ and
$\sphericalangle\go P_{1,2}\go P_{2,3}\dach{\go F}_{1,4}$, where the latter equals   $\sphericalangle\go P_{3,4}\go P_{1,4}\go F_{1,4}$ due to the line-symmetry (see Fig.\ \ref{fig5}). Therefore we can set up the following equation:
\begin{equation}\label{eq:delto1}
    \left\langle
    \frac{\Vkt P_{1,2}-\Vkt P_{2,3}}{\| \Vkt P_{1,2}-\Vkt P_{2,3}\|},\Vkt r_{2,3}
    \right\rangle
    -
    \left\langle
    \frac{\Vkt P_{3,4}-\Vkt P_{1,4}}{\| \Vkt P_{3,4}-\Vkt P_{1,4}\|},\Vkt r_{1,4}
    \right\rangle=0.
\end{equation}
Now we consider the other pair of adjacent angles $\sphericalangle\go P_{3,4}\go P_{2,3}\go F_{2,3}$ and
$\sphericalangle\go P_{3,4}\go P_{2,3}\dach{\go F}_{1,4}$, where the latter equals   $\sphericalangle\go P_{1,2}\go P_{1,4}\go F_{1,4}$ due to the line-symmetry (see Fig.\ \ref{fig5}). Therefore we can set up the following equation:
\begin{equation}\label{eq:delto2}
    \left\langle
    \frac{\Vkt P_{3,4}-\Vkt P_{2,3}}{\| \Vkt P_{3,4}-\Vkt P_{2,3}\|},\Vkt r_{2,3}
    \right\rangle
    -
    \left\langle
    \frac{\Vkt P_{1,2}-\Vkt P_{1,4}}{\| \Vkt P_{1,2}-\Vkt P_{1,4}\|},\Vkt r_{1,4}
    \right\rangle=0.
\end{equation}
Note that using the equalities of Eq.\ (\ref{eq:linsym}) renders the numerators of the Eqs.\ (\ref{eq:delto1}) and (\ref{eq:delto2}) algebraic, which read as follows: 
\begin{equation*}
    \begin{split}
    &(\mu_{1,4} - \mu_{1,2} - \mu_{2,3} + \mu_{3,4})a_2 + \mu_{1,4} + \mu_{1,2} - \mu_{2,3} - \mu_{3,4}=0, \\
    &(\mu_{1,4} + \mu_{1,2} - \mu_{2,3} - \mu_{3,4})a_1 +\mu_{1,4} - \mu_{1,2} - \mu_{2,3} + \mu_{3,4}=0.
    \end{split}    
\end{equation*}
For family (B) these two equations are fulfilled identically. \hfill $\BewEnde$
\end{proof}

%%%%%%%%%%%%%%%%%%%%%%%%%%%%%%%%%%%%%%%%%%%%%%%%%%%%%%%%%%%%%%%%%%%%%%%%%%%%%%%%%%%%%%%%%%%%%%%%%%%%%%%%%%%%%%%%%%%%%%%%%%%%%

\subsection{Non-symmetric arrangements}\label{subsec:nonsym}

In this section we do not assume that the two Bennetts  $\mathcal{B}$ and $\dach{\mathcal{B}}$ are linked by a symmetry relation. A necessary and sufficient set of algebraic conditions for such a general (non-symmetric) arrangement can be obtained by the following distance geometric approach: 

Let us consider the Bennett $\mathcal{B}$ given in Section \ref{subsec:sym}, where the location of the point $\go P_{i,j}$ depends on the motion parameter $\tau$ and the design parameters $a_1,a_2,k,\mu_{i,j}$. 

Let  us assume that we have a second Bennett  $\bar{\mathcal{B}}$ which is parametrized and coordinatized analogously; i.e.\ the location of the points $\bar{\go P}_{i,j}$ depends on the motion parameter $\bar \tau$ and the design parameters $\bar a_1,\bar a_2,\bar k,\bar \mu_{i,j}$. 

The two  Bennetts  $\mathcal{B}$ and $\bar{\mathcal{B}}$ can be flexible coupled over the two  tetrahedra $\go P_{1,4},\go P_{1,2}, \go P_{2,3}, \go P_{3,4}$ and $\bar{\go P}_{1,4},\bar{\go P}_{1,2}, \bar{\go P}_{2,3}, \bar{\go P}_{3,4}$ if for each $\tau$ there exists a $\bar\tau$ such that the two tetrahedra are isometric. Then the two tetrahedra can either by identified by an orientation-preserving or orientation-reversing isometry $\delta$; where $\dach{\mathcal{B}}:=\delta(\bar{\mathcal{B}})$. The algebraic conditions for this isometry read as follows:
\begin{equation}\label{eq:4ohne}
\begin{split}
&\|\Vkt P_{1,4}-\Vkt P_{1,2} \|^2 -  \|\bar{\Vkt P}_{1,4}-\bar{\Vkt P}_{1,2} \|^2=0 \\
&\|\Vkt P_{1,2}-\Vkt P_{2,3} \|^2 -  \|\bar{\Vkt P}_{1,2}-\bar{\Vkt P}_{2,3} \|^2=0 \\
&\|\Vkt P_{2,3}-\Vkt P_{3,4} \|^2 -  \|\bar{\Vkt P}_{2,3}-\bar{\Vkt P}_{3,4} \|^2=0 \\
&\|\Vkt P_{3,4}-\Vkt P_{1,4} \|^2 -  \|\bar{\Vkt P}_{3,4}-\bar{\Vkt P}_{1,4} \|^2=0 
\end{split}
\end{equation}
and
\begin{equation}\label{eq:2mit}
\begin{split}
&\|\Vkt P_{1,4}-\Vkt P_{2,3} \|^2 -  \|\bar{\Vkt P}_{1,4}-\bar{\Vkt P}_{2,3} \|^2=0 \\
&\|\Vkt P_{1,2}-\Vkt P_{3,4} \|^2 -  \|\bar{\Vkt P}_{1,2}-\bar{\Vkt P}_{3,4} \|^2=0.
\end{split}
\end{equation}
Note that only the two conditions of Eq.\ (\ref{eq:2mit}) depend on the motion parameters $\tau$ and $\bar\tau$. They are of the form: 
\begin{equation*}
c_{22}\tau^2\bar\tau^2+ 
c_{21}\tau^2\bar\tau+ 
c_{12}\tau\bar\tau^2+ 
c_{20}\tau^2+ 
c_{02}\bar\tau^2+ 
c_{10}\tau+ 
c_{01}\bar\tau+ 
c_{00}=0. 
\end{equation*}
Their resultant with respect to $\bar\tau$ yields an equation 
of degree 8 in $\tau$; i.e.\ 
\begin{equation*}
c_8\tau^8+c_7\tau^7+c_6\tau^6+c_5\tau^5+c_4\tau^4+c_3\tau^3+c_2\tau^2+c_1\tau+c_0=0
\end{equation*}
where the $c_i$'s depend on the 14 unknowns 
\begin{equation*}
a_1,a_2,k,\mu_{1,4},\mu_{1,2},\mu_{2,3},\mu_{3,4},\bar a_1,\bar a_2,\bar k,\bar \mu_{1,4},\bar \mu_{1,2},\bar \mu_{2,3},\bar \mu_{3,4}.
\end{equation*}
This equation has to be fulfilled independent of $\tau$ implying the nine conditions $c_8=c_7=\ldots=c_0=0$. Together with the four conditions of Eq.\ (\ref{eq:4ohne}) we have 13 equations in 14 unknowns, but we have not to forget that the scaling factor is not fixed yet. Therefore we can set e.g.\ $k=1$. This gives us now a square system of equations and all its solutions\footnote{Clearly, due to the generality also the solutions of Section \ref{subsec:sym} are contained within this solution set as well as trivial solutions $\mathcal{B}=\dach{\mathcal{B}}$.} correspond to flexible arrangements of two Bennetts.

We were not able to solve this system of necessary and sufficient conditions due to its complexity. 
However, we succeeded in finding one further family (C)
which is described in the next theorem.

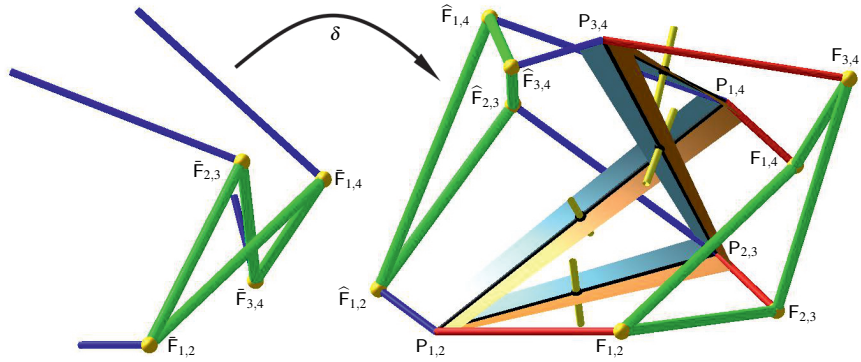
\begin{figure}[t]
\begin{center}
\begin{overpic}
   [height=47mm]{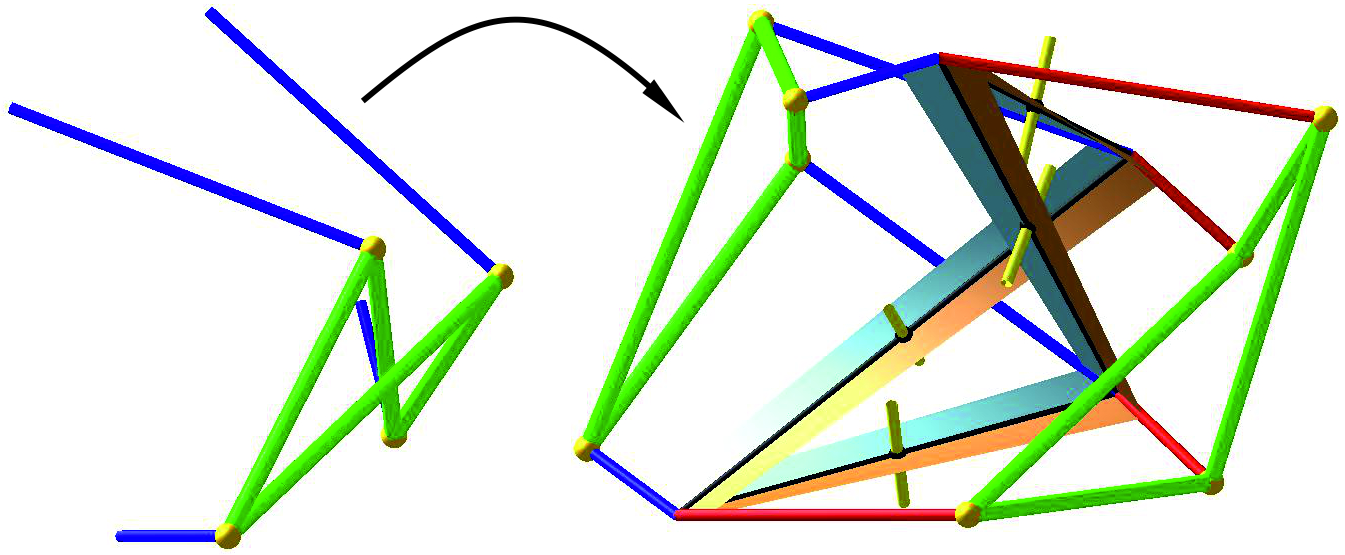}
\begin{scriptsize}
\put(19,0.5){$\bar{\go F}_{1,2}$}
\put(38.5,20){$\bar{\go F}_{1,4}$}
\put(27,6){$\bar{\go F}_{3,4}$}
\put(22,21){$\bar{\go F}_{2,3}$}
\put(38,37){$\delta$}
\put(39,6){$\dach{\go F}_{1,2}$}
\put(51,39){$\dach{\go F}_{1,4}$}
\put(54.5,29.5){$\dach{\go F}_{2,3}$}
\put(60.5,31.5){$\dach{\go F}_{3,4}$}
\put(67,38.5){$\go P_{3,4}$}
\put(48,0.5){$\go P_{1,2}$}
\put(83,31){$\go P_{1,4}$}
\put(84.5,12){$\go P_{2,3}$}
\put(96.5,34.5){$\go F_{3,4}$}
\put(87,22.5){$\go F_{1,4}$}
\put(91.5,4.5){$\go F_{2,3}$}
\put(69,0.5){$\go F_{1,2}$}
\end{scriptsize}     
  \end{overpic} 
\end{center}	
\caption{
The two  Bennetts  $\mathcal{B}$ and $\bar{\mathcal{B}}$ can be flexible coupled by on orientation-preserving isometry $\delta$ with $\dach{\mathcal{B}}:=\delta(\bar{\mathcal{B}})$. 
The resulting flexible bi-Bennett belongs to family (C) of Theorem \ref{thm:nonsym}. We only labeled the points (but not the rotation axes) to avoid an overloading of the figure. In addition the axes of the half-turns described in Theorem \ref{thm:property2} are displayed in yellow. This illustration corresponds to the following parameters: $k=1$, $s=1$ $\mu_{2,3}=\tfrac{2}{3}$, $\mu_{3,4}=\tfrac{1}{4}$, $a_1=\tfrac{1}{2}$, $a_2=\tfrac{1}{3}$, $\tau=\tfrac{9}{10}$ and 
$\bar\tau=-\tfrac{\sqrt{195165444215}}{357245}$.}
  \label{fig6}
\end{figure}

\begin{theorem}\label{thm:nonsym}
    Besides the factor of similarity there exists a 4-dimensional family (C) of non-symmetric flexible arrangements of two Bennett tubes. Under consideration of $s\in\left\{-1,+1\right\}$ this family (C) is given by:  
    \begin{equation*}
    \begin{split}
    &\bar a_1= a_1, \quad 
    \bar a_2= a_2, \quad
    \mu_{2,3}=\mu_{1,4}, \quad 
     \mu_{3,4}=\mu_{1,2},\quad \bar k= k \\  
     &\bar \mu_{1,4}=s\mu_{1,2}, \quad
    \bar \mu_{1,2}=s\mu_{1,4}, \quad
    \bar \mu_{2,3}=s\mu_{1,2}, \quad
     \bar \mu_{3,4}=s\mu_{1,4}.
    \end{split}
    \end{equation*}
\end{theorem}

\begin{proof}
We eliminate the factor of similarity by setting $k=1$, thus we remain with the four free parameters $a_1,a_2,\mu_{1,2},\mu_{1,4}$. 
It can easily be checked that the four conditions of Eq.\ (\ref{eq:4ohne}) are fulfilled identically. Moreover, it turns out that the remaining two conditions of Eq.\ (\ref{eq:2mit}) are identical and split up into $(a_1^2+1)(a_2+1)U$ where the factor $U$ equals:
\begin{equation*}
\begin{split}
    &(\mu_{1,4}^2-\mu_{1,2}^2)(a_1-a_2)^2\tau^2\bar{\tau}^2 + 
    \left[(\mu_{1,4}^2-\mu_{1,2}^2)(a_1^2+a_2^2)+2(\mu_{1,4}^2+\mu_{1,2}^2+2)a_1a_2\right]\tau^2 \\
     &\left[(\mu_{1,4}^2-\mu_{1,2}^2)(a_1^2+a_2^2)-2(\mu_{1,4}^2+\mu_{1,2}^2+2)a_1a_2\right]\bar{\tau}^2 + (\mu_{1,4}^2-\mu_{1,2}^2)(a_1+a_2)^2
\end{split}    
\end{equation*}
for both options of $s$. This already proves that $\mathcal{B}$ and $\bar{\mathcal{B}}$ can be flexible arranged.
\hfill $\BewEnde$
\end{proof}

Note that in addition to the option $s\in\left\{-1,+1\right\}$ equation $U=0$ has two solutions for $\bar\tau$. Thus in total one can distinguish four branches. For all four branches the tetrahedra $\go P_{1,4},\go P_{1,2}, \go P_{2,3}, \go P_{3,4}$ and $\bar{\go P}_{1,4},\bar{\go P}_{1,2}, \bar{\go P}_{2,3}, \bar{\go P}_{3,4}$ have the same orientation, which can easily be checked (e.g.\ by using Maple) verifying the condition:
\begin{equation*}
    \det (\Vkt  P_{1,2} -\Vkt P_{1,4},
    \Vkt  P_{2,3} -\Vkt P_{1,4},
    \Vkt  P_{3,4} -\Vkt P_{1,4})-
    \det (\bar{\Vkt  P}_{1,2} -\bar{\Vkt P}_{1,4},
    \bar{\Vkt  P}_{2,3} -\bar{\Vkt P}_{1,4},
    \bar{\Vkt  P}_{3,4} -\bar{\Vkt P}_{1,4})=0
\end{equation*}
Therefore, one can always use an rigid-body motion $\delta$ to bring the two Bennetts in place; i.e.\ $\dach{\mathcal{B}}:=\delta(\bar{\mathcal{B}})$ with 
$\go P_{i,j}=\dach {\go P}_{i,j} =\delta(\bar{\go P}_{i,j})$.

\begin{theorem}\label{thm:property2}
The spherical indicatrices of the spherical 4R-loops around opposite centers $\go P_{i,j}$ of family (C) are related by a direct isometry. 
The spherical indicatrices of the spherical 4R-loops around adjacent centers $\go P_{i,j}$ of family (C) correspond to two motion modes of the same spherical 4-bar. Moreover, two adjacent vertices $\go P_{i,j}$ and $\go P_{j,k}$ are related by a half-turn $\rho$ with:
$\go P_{i,j}\mapsto \go P_{j,k}$
$\go F_{i,j}\mapsto \dach{\go F}_{j,k}$
$\dach{\go F}_{i,j}\mapsto \go F_{j,k}$.
\end{theorem}

\begin{proof}
We do the proof exemplarily for the adjacent vertices\footnote{For all other vertices the proof is the same, there is just a shift of indices.} $\go P_{2,3}$ and $\go P_{3,4}$ and split it up into three steps: \medskip

\noindent
{\bf Step 1)} In a first step we show equality relations between angles spanned by adjacent rotary joints 
around the vertices $\go P_{2,3}$ and $\go P_{3,4}$. As these angles do not depend on the motion parameter $\tau$ there is no need to solve $U=0$ for $\bar\tau$. 
\begin{enumerate}
    \item 
 We want to show that  $\sphericalangle\go P_{3,4}\go P_{2,3}\go F_{2,3}$ equals  
 $\sphericalangle\go P_{2,3}\go P_{3,4}\dach{\go F}_{3,4}$. As the latter angle equals  
 $\sphericalangle\bar {\go P}_{2,3}\bar{\go P}_{3,4}\bar{\go F}_{3,4}$, we only have to check the equation
\begin{equation}\label{eq:gleich1}
    \left\langle
    \Vkt P_{3,4}-\Vkt P_{2,3},
    \Vkt F_{2,3}-\Vkt P_{2,3}
    \right\rangle
    -
    \left\langle
    \bar{\Vkt P}_{2,3}-\bar{\Vkt P}_{3,4},
    \bar{\Vkt F}_{3,4}-\bar{\Vkt P}_{3,4}
    \right\rangle
    =0
\end{equation}
as $\| \Vkt P_{3,4}-\Vkt P_{2,3}\|=\|\bar{\Vkt P}_{2,3}-\bar{\Vkt P}_{3,4}\|$ and 
$\|\Vkt F_{2,3}-\Vkt P_{2,3}\|=\| \bar{\Vkt F}_{3,4}-\bar{\Vkt P}_{3,4} \|$
hold true. 
\item 
Further, we want to show that  $\sphericalangle\go P_{1,2}\go P_{2,3}\go F_{2,3}$ equals  
 $\sphericalangle\go P_{1,4}\go P_{3,4}\dach{\go F}_{3,4}$. As the latter angle equals  
 $\sphericalangle\bar {\go P}_{1,4}\bar{\go P}_{3,4}\bar{\go F}_{3,4}$, we only have to check the equation
\begin{equation}\label{eq:gleich2}
    \left\langle
    \Vkt P_{1,2}-\Vkt P_{2,3},
    \Vkt F_{2,3}-\Vkt P_{2,3}
    \right\rangle
    -
    \left\langle
    \bar{\Vkt P}_{1,4}-\bar{\Vkt P}_{3,4},
    \bar{\Vkt F}_{3,4}-\bar{\Vkt P}_{3,4}
    \right\rangle
    =0
\end{equation}
as $\| \Vkt P_{1,2}-\Vkt P_{2,3}\|=\|\bar{\Vkt P}_{1,4}-\bar{\Vkt P}_{3,4}\|$ and 
$\|\Vkt F_{2,3}-\Vkt P_{2,3}\|=\| \bar{\Vkt F}_{3,4}-\bar{\Vkt P}_{3,4} \|$
hold true. 
 \item 
 We also want to show that $\sphericalangle\go P_{2,3}\go P_{3,4}\go F_{3,4}$ equals $\sphericalangle\go P_{3,4}\go P_{2,3}\dach{\go F}_{2,3}$. As the latter angle equals  
 $\sphericalangle\bar {\go P}_{3,4}\bar{\go P}_{2,3}\bar{\go F}_{2,3}$, we only have to check the equation
\begin{equation}\label{eq:gleich3}
    \left\langle
    \Vkt P_{2,3}-\Vkt P_{3,4},
    \Vkt F_{3,4}-\Vkt P_{3,4}
    \right\rangle
    -
    \left\langle
    \bar{\Vkt P}_{3,4}-\bar{\Vkt P}_{2,3},
    \bar{\Vkt F}_{2,3}-\bar{\Vkt P}_{2,3}
    \right\rangle
    =0
\end{equation}
as $\| \Vkt P_{2,3}-\Vkt P_{3,4}\|=\|\bar{\Vkt P}_{3,4}-\bar{\Vkt P}_{2,3}\|$ and 
$\|\Vkt F_{3,4}-\Vkt P_{3,4}\|=\| \bar{\Vkt F}_{2,3}-\bar{\Vkt P}_{2,3} \|$
hold true. 
\item Finally, we want to show that  
$\sphericalangle\go P_{1,4}\go P_{3,4}{\go F}_{3,4}$ equals
$\sphericalangle\go P_{1,2}\go P_{2,3}\dach{\go F}_{2,3}$. As the latter angle equals  
 $\sphericalangle\bar {\go P}_{1,2}\bar{\go P}_{2,3}\bar{\go F}_{2,3}$, we only have to check the equation
\begin{equation}\label{eq:gleich4}
    \left\langle
    \Vkt P_{1,4}-\Vkt P_{3,4},
    \Vkt F_{3,4}-\Vkt P_{3,4}
    \right\rangle
    -
    \left\langle
    \bar{\Vkt P}_{1,2}-\bar{\Vkt P}_{2,3},
    \bar{\Vkt F}_{2,3}-\bar{\Vkt P}_{2,3}
    \right\rangle
    =0
\end{equation}
as $\| \Vkt P_{1,4}-\Vkt P_{2,3}\|=\|\bar{\Vkt P}_{1,2}-\bar{\Vkt P}_{2,3}\|$ and 
$\|\Vkt F_{3,4}-\Vkt P_{3,4}\|=\| \bar{\Vkt F}_{2,3}-\bar{\Vkt P}_{2,3} \|$
hold true. 
\end{enumerate}
Eqs.\ (\ref{eq:gleich1}--\ref{eq:gleich4}) can easily be verified using e.g.\ Maple. \medskip 

\noindent
{\bf Step 2)}
In this step we prove that the angles 
$\sphericalangle\go F_{2,3}\go P_{2,3}\dach{\go F}_{2,3}$ and 
$\sphericalangle\go F_{3,4}\go P_{3,4}\dach{\go F}_{3,4}$
are equal. This is a little more complicated as these angles are enclosed by opposite rotary angles of the 4R loops around the vertices $\go P_{2,3}$ and $\go P_{3,4}$. Therefore, we cannot reduce the problem to $\mathcal{B}$ and 
$\bar{\mathcal{B}}$ but has to consider the spatial arrangement  $\mathcal{B}$ and  $\dach{\mathcal{B}}$, which also depends on the motion parameter $\tau$. 
For the proof we do not solve  $U=0$ for $\bar\tau$ but we take this condition in a later phase of the proof into account. 

As the two tetrahedra $\go P_{1,4},\go P_{1,2}, \go P_{2,3}, \go P_{3,4}$ and $\bar{\go P}_{1,4},\bar{\go P}_{1,2}, \bar{\go P}_{2,3}, \bar{\go P}_{3,4}$ are related by a direct isometry we can compute $\dach{\Vkt F}_{2,3}$ and  $\dach{\Vkt F}_{3,4}$ as follows:
\begin{enumerate}
    \item 
We express $\bar{\Vkt F}_{2,3}$ as a linear combination of the following form:
\begin{equation*}
\bar{\Vkt F}_{2,3}=\bar{\Vkt P}_{2,3} + 
\xi_{2,3}(\bar{\Vkt P}_{3,4}-\bar{\Vkt P}_{2,3})+
\eta_{2,3}(\bar{\Vkt P}_{1,2}-\bar{\Vkt P}_{2,3})+
\zeta_{2,3}(\bar{\Vkt P}_{1,4}-\bar{\Vkt P}_{3,4})
\end{equation*}
The corresponding three equations on the coordinate level can be solved for $\xi_{2,3},\eta_{2,3},\zeta_{2,3}$. Based on the obtained values we can compute 
$\dach{\Vkt F}_{2,3}$ as follows:
\begin{equation}\label{eq:fdach23}
\dach{\Vkt F}_{2,3}={\Vkt P}_{2,3} + 
\xi_{2,3}({\Vkt P}_{3,4}-{\Vkt P}_{2,3})+
\eta_{2,3}({\Vkt P}_{1,2}-{\Vkt P}_{2,3})+
\zeta_{2,3}({\Vkt P}_{1,4}-{\Vkt P}_{3,4})
\end{equation}
\item 
We express $\bar{\Vkt F}_{3,4}$ as a linear combination of the following form:
\begin{equation*}
\bar{\Vkt F}_{3,4}=\bar{\Vkt P}_{2,3} + 
\xi_{3,4}(\bar{\Vkt P}_{2,3}-\bar{\Vkt P}_{3,4})+
\eta_{3,4}(\bar{\Vkt P}_{1,4}-\bar{\Vkt P}_{3,4})+
\zeta_{3,4}(\bar{\Vkt P}_{1,2}-\bar{\Vkt P}_{2,3})
\end{equation*}
The corresponding three equations on the coordinate level can be solved for $\xi_{3,4},\eta_{3,4},\zeta_{3,4}$. Based on the obtained values we can compute 
$\dach{\Vkt F}_{3,4}$ as follows:
\begin{equation}\label{eq:fdach34}
\dach{\Vkt F}_{3,4}={\Vkt P}_{2,3} + 
\xi_{3,4}({\Vkt P}_{2,3}-{\Vkt P}_{3,4})+
\eta_{3,4}({\Vkt P}_{1,4}-{\Vkt P}_{3,4})+
\zeta_{3,4}({\Vkt P}_{1,2}-{\Vkt P}_{2,3})
\end{equation}
\end{enumerate}
Using the expressions of Eq.\ (\ref{eq:fdach23}) and 
Eq.\ (\ref{eq:fdach34}) we can prove the equality of the angles $\sphericalangle\go F_{2,3}\go P_{2,3}\dach{\go F}_{2,3}$ and 
$\sphericalangle\go F_{3,4}\go P_{3,4}\dach{\go F}_{3,4}$ by computing
\begin{equation}\label{eq:diag}
    \left\langle
    \Vkt F_{2,3}-\Vkt P_{2,3},
    \dach{\Vkt F}_{2,3}-\Vkt P_{2,3}
    \right\rangle
    -
    \left\langle
    \Vkt F_{3,4}-\Vkt P_{3,4},
    \dach{\Vkt F}_{3,4}-\Vkt P_{3,4}
    \right\rangle
    =0
\end{equation}
as $\| \Vkt F_{2,3}-\Vkt P_{2,3} \|=\|\dach{\Vkt F}_{3,4}-\Vkt P_{3,4} \|$ and 
$\|\dach{\Vkt F}_{2,3}-\Vkt P_{2,3}\|=\| \Vkt F_{3,4}-\Vkt P_{3,4} \|$
hold true. Now we factor Eq.\ (\ref{eq:diag}) by using Maple and see that $U$ appears as one of its factors, which finishes the second step of the proof. \medskip

\noindent
{\bf Step 3)}
Due to the relations 
\begin{equation*}
\| \Vkt F_{2,3}-\Vkt P_{2,3} \|=\|\dach{\Vkt F}_{3,4}-\Vkt P_{3,4} \|, 
\qquad
\| \dach{\Vkt F}_{2,3}-\Vkt P_{2,3} \|=\|{\Vkt F}_{3,4}-\Vkt P_{3,4} \|, 
\end{equation*}
and the angle equalities  shown in steps 1 and 2, the two tetrahedra $\go P_{2,3},\go P_{3,4},\go F_{2,3},\dach {\go F}_{2,3}$ and 
$\go P_{3,4},\go P_{2,3},\dach{\go F}_{3,4},{\go F}_{3,4}$ 
have to be isometric. In addition, they have the same orientation, which can be checked as follows. We compute 
\begin{equation*}
    \det (\Vkt  P_{3,4} -\Vkt P_{2,3},
    \Vkt  F_{2,3} -\Vkt P_{2,3},
    \dach{\Vkt  F}_{2,3} -\Vkt P_{2,3})
    -
    \det (\Vkt  P_{2,3} -\Vkt P_{3,4},
    \dach{\Vkt  F}_{3,4} -\Vkt P_{3,4},
    \Vkt  F_{3,4} -\Vkt P_{3,4})=0
\end{equation*}
and see again that one factor of it equals $U$. 
This already implies the existence of the half-turn $\rho$ 
 with:
$\go P_{2,3}\mapsto \go P_{3,4}$
$\go F_{2,3}\mapsto \dach{\go F}_{3,4}$
$\dach{\go F}_{2,3}\mapsto \go F_{3,4}$.
But this half-turn does not map $\go P_{1,2}$ to 
$\go P_{1,4}$ because this would imply the equality of the angles 
$\sphericalangle\go P_{1,2}\go P_{2,3}{\go P}_{3,4}$ and 
$\sphericalangle\go P_{1,4}\go P_{3,4}{\go P}_{2,3}$. But this is in general not the case as it can easily be checked that  
\begin{equation*}
   \left\langle
    \Vkt P_{1,2}-\Vkt P_{2,3},
    \Vkt P_{3,4}-\Vkt P_{2,3}
    \right\rangle
    -
    \left\langle
     \Vkt P_{1,4}-\Vkt P_{3,4},
    \Vkt P_{2,3}-\Vkt P_{3,4}
    \right\rangle
    \neq 0
\end{equation*}
under consideration of $U=0$. Note that $\rho(\go P_{1,2})$ 
and $\go P_{1,4}$ are related by a reflection along the plane spanned by $\go P_{3,4},\go F_{3,4},\dach{\go F}_{3,4}$.

This shows that the spherical indicatrices of the spherical 4R-loops around adjacent centers are the same spherical 4-bars but in different motion modes. 
The propagation of this property implies that the spherical indicatrices of the spherical 4R-loops around opposite centers are directly isometric. \medskip

The above three steps hold true for both possible values for $s$.  \hfill $\BewEnde$
\end{proof}

\begin{remark}\label{rem:6R}
Note that within each flexible arrangement of two Bennett tubes one can identify four different 6R loops (see Fig.\ \ref{fig7}). 
For the flexible bi-Bennetts of family (A) and (B) these overconstrained 6R linkages belong to the well-known class of line-symmetric 6R loops. But for the 6R loops contained in bi-Bennetts of family (C) it remains open whether they are novel or not.  \hfill $\diamond$
\end{remark}

\begin{figure}[t]
\begin{center}
\begin{overpic}
   [height=47mm]{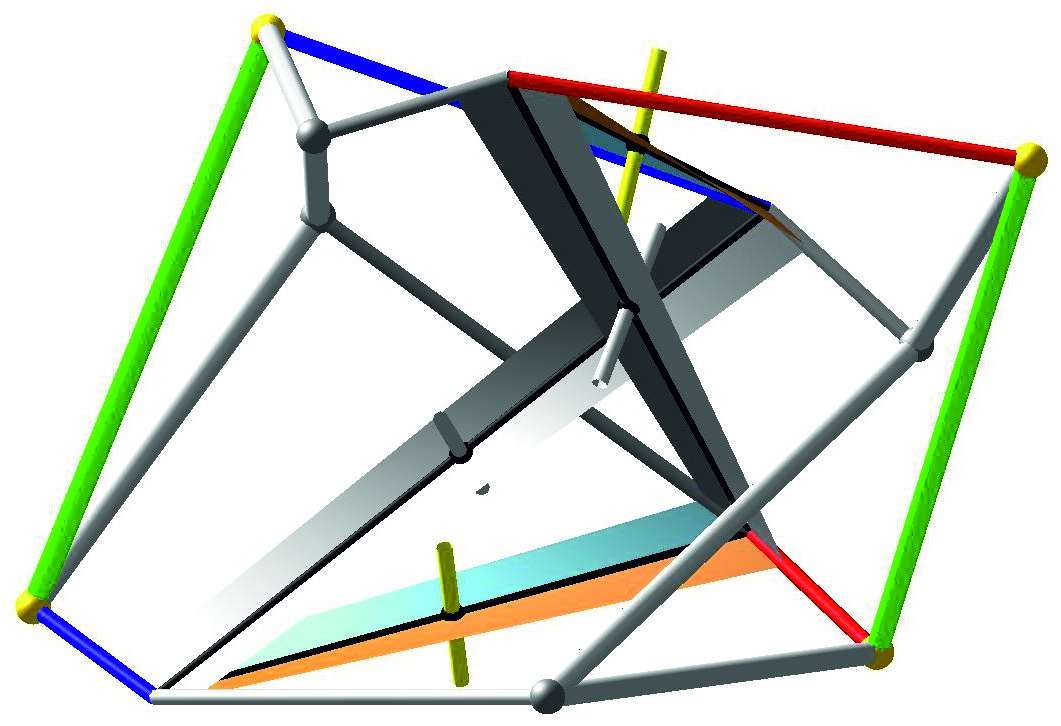}    
  \end{overpic} 
\end{center}	
\caption{One 6R loop is highlighted, which is contained in an flexible bi-Bennett of family (C).}
  \label{fig7}
\end{figure}

%%%%%%%%%%%%%%%%%%%%%%%%%%%%%

\section{Prismatic and pyramidal limits of  flexible bi-Bennetts} \label{sec:limits_results}

In this section we study the prismatic and pyramidal limits of the obtained families of flexible bi-Bennetts to shed some light on their relation to the known classes of flexible biprisms and bipyramids (see Subsection \ref{sec:review}), respectively. 

\subsection{Prismatic limits of families (A) and (B)}

Due to the line-symmetric construction of the families (A) and (B) also their prismatic limits have this property. Therefore, they belong to the class (III1) of Subsection \ref{sec:review}.

As pointed out in Subsection \ref{sec:planar} 
the anti-parallelogramic/parallelogrammic prismatic limit can always be obtained in two ways. Without loss of generality, we can assume  case (2a) for the anti-parallelogramic limit and case (2b) for the parallelogramic limit. Using this convention we get the following results:

\subsubsection{Anti-parallelogramic prisms}

In this case, Eq.\ (\ref{eq:linsym}) simplifies to 
\begin{equation*}
\begin{split}
&(\mu_{1,4} - \mu_{1,2} + \mu_{2,3} - \mu_{3,4})(\mu_{1,4} - \mu_{1,2} - \mu_{2,3} + \mu_{3,4})=0
\\
&(\mu_{1,4} - \mu_{1,2} + \mu_{2,3} - \mu_{3,4})
(\mu_{1,4} + \mu_{1,2} - \mu_{2,3} - \mu_{3,4})=0
\end{split}
\end{equation*}
We can distinguish two solution:
\begin{enumerate}
    \item One solution equals Eq.\ (\ref{eq1}) and it corresponds to family (B). Moreover, the obtained flexible arrangements of two anti-paralleogramic prisms have the same property as family (B) mentioned in Theorem \ref{thm:property}. In this case, the points
    $\go P_{1,4},\go P_{1,2}, \go P_{2,3}, \go P_{3,4}$ are located in a plane which is orthogonal to the symmetry plane of the anti-paralleogramic prisms. The intersection line of these two planes equals the line of symmetry. Therefore, these structures belong in addition to the class (III2ii) of Subsection \ref{sec:review}. 
    \item 
    $\mu_{1,4} = \mu_{1,2} - \mu_{2,3} + \mu_{3,4}$: This solution set 
    contains more flexible arrangements of two anti-paralleogramic prisms than only the ones corresponding to family (A), as a generic element of this set does not possess the isogonality property mentioned in  Theorem \ref{thm:property}. For meeting this property the following additional condition has to hold:
\begin{equation*}
(d_1\mu_{1,2} - d_1\mu_{2,3} - d_2\mu_{2,3} + d_2\mu_{3,4})
(d_1\mu_{1,2} - d_1\mu_{2,3} + d_2\mu_{2,3} - d_2\mu_{3,4})=0
\end{equation*}
In this case the structure has two flat-poses thus it also belongs to the class (III3). 
\end{enumerate}

\subsubsection{Parallelogramic prisms}

In this case Eq.\ (\ref{eq:linsym}) simplifies to 
\begin{equation*}
\begin{split}
&(\mu_{1,4} - \mu_{1,2} - \mu_{2,3} + \mu_{3,4})(\mu_{1,4} - \mu_{1,2} + \mu_{2,3} - \mu_{3,4})=0
\\
&(\mu_{1,4} - \mu_{1,2} - \mu_{2,3} + \mu_{3,4})
(\mu_{1,4} + \mu_{1,2} + \mu_{2,3} + \mu_{3,4})=0
\end{split}
\end{equation*}
We can distinguish two solution:
\begin{enumerate}
    \item One solution equals Eq.\ (\ref{eq2}), thus  the line of symmetry of $\go P_{1,4},\go P_{1,2}, \go P_{2,3}, \go P_{3,4}$ coincides with the one of $\go F_{1,4},\go F_{1,2}, \go F_{2,3}, \go F_{3,4}$. 
As a consequence, the prism will be mapped onto itself using this line of symmetry, which is a trivial solution to our problem.
\item 
$\mu_{1,4} = \mu_{1,2} + \mu_{2,3} - \mu_{3,4}$: Now  $\go P_{1,4},\go P_{1,2}, \go P_{2,3}, \go P_{3,4}$ form again a parallelogram; i.e.\ it is a planar cut of the prism. Therefore this solution set, which also belongs to the class (III4ii), 
    contains more flexible arrangements of two paralleogramic prisms than only the ones corresponding to family (A), as a generic element of this set does not possess the isogonality property mentioned in  Theorem \ref{thm:property}. For meeting this property the following additional condition has to hold:
\begin{equation*}
(d_1\mu_{1,2} + d_1\mu_{2,3} + d_2\mu_{2,3} - d_2\mu_{3,4})
(d_1\mu_{1,2} + d_1\mu_{2,3} - d_2\mu_{2,3} + d_2\mu_{3,4})=0
\end{equation*}
\end{enumerate}

\subsection{Pyramidal limits of families (A) and (B)}

Here we get exactly the same conditions as in the Bennett case but for $k=0$. 
Therefore, the pyramidal limit of family (A) belongs to the intersection of the classes (I1) and (I3). The pyramidal limit of family (B) belongs to the intersection of the classes (I1) and (I2).

\begin{figure}[t]
\begin{center}
\begin{overpic}
   [height=50mm]{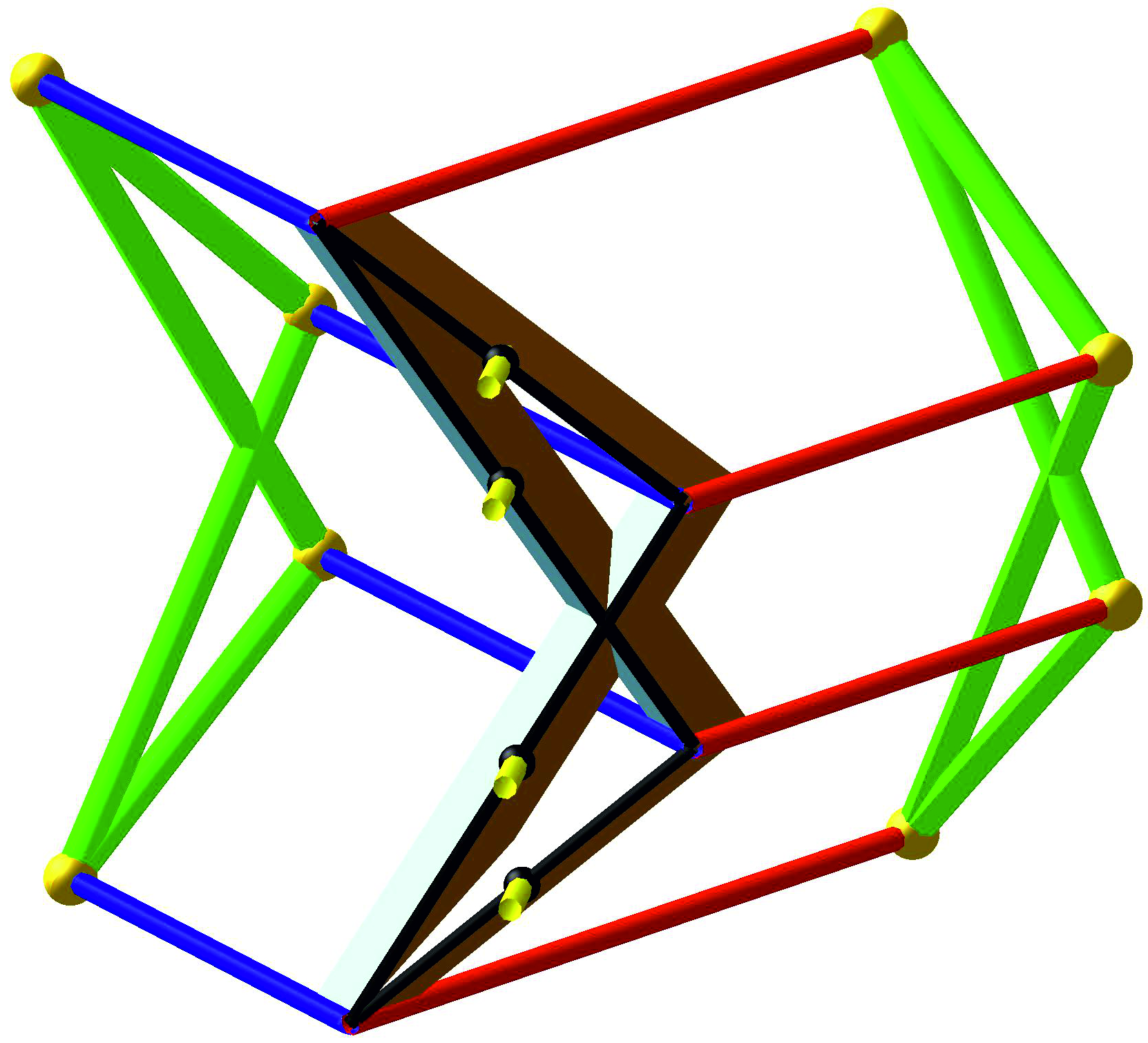}
\begin{scriptsize}
\put(0,0){a)}
\end{scriptsize}     
  \end{overpic} 
\hfill
\begin{overpic}
[height=38mm]{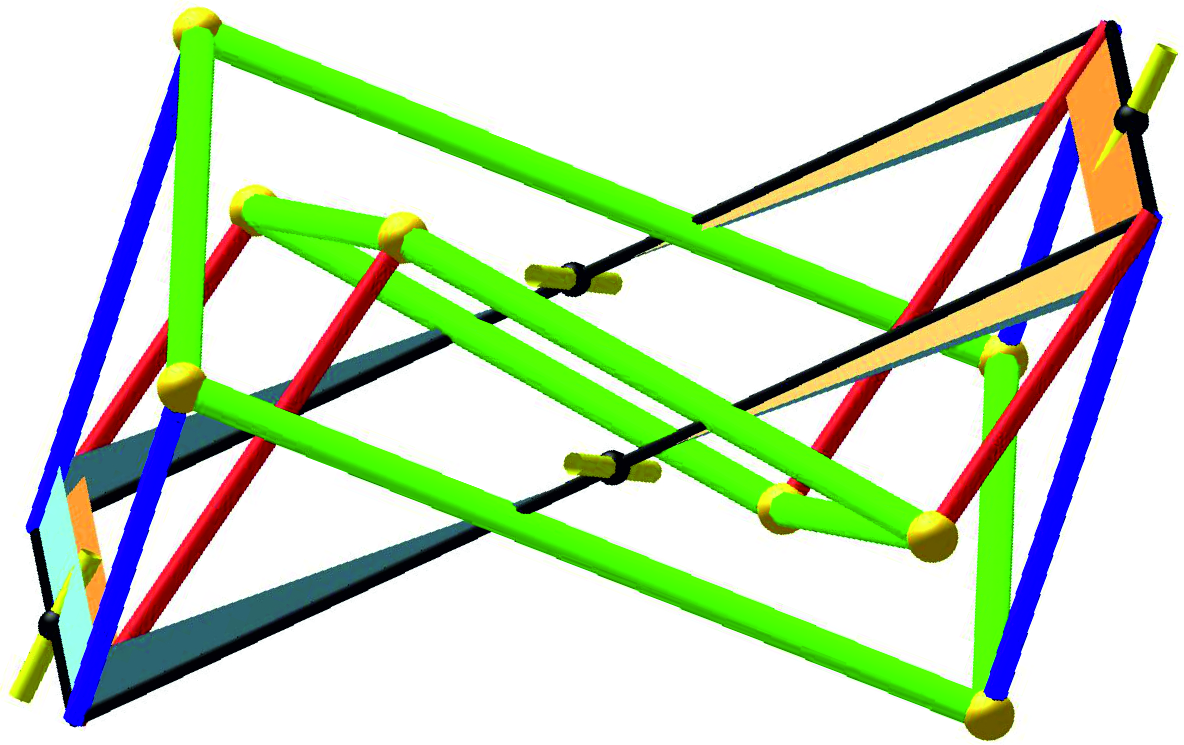}
\begin{scriptsize}
\put(-3,0){b)}
\end{scriptsize}  
\end{overpic}
\end{center}	
\caption{Illustration of the prismatic limits of family (C) for the anti-parallelogramic case (a) and the parallelogramic one (b). 
Illustration (a) corresponds to the following parameters: 
$k=1$, $s=1$ $\mu_{2,3}=\tfrac{2}{3}$, $\mu_{3,4}=\tfrac{1}{2}$, $d_1=\tfrac{1}{2}$, $d_2=\tfrac{1}{3}$, $\tau=\tfrac{3}{4}$ and 
$\bar\tau=\tfrac{3}{4}$.
 Illustration (b) corresponds to the following parameters: 
$k=1$, $s=1$ $\mu_{2,3}=\tfrac{1}{3}$, $\mu_{3,4}=\tfrac{1}{2}$, $d_1=\tfrac{2}{3}$, $d_2=\tfrac{3}{4}$, $\tau=\tfrac{3}{4}$ and 
$\bar\tau=-\tfrac{\sqrt{15281}}{413}$.
}
  \label{fig8}
\end{figure}    

\subsection{Prismatic limits of family (C)}

The computation of both (anti-parallelogramic and parallelogramic) prismatic limits is straight forward using the cases (2a) and (2b), respectively, of Subsection \ref{sec:planar}. 
Only the four equations given in Eq.\ (\ref{eq:4ohne}) 
are not fulfilled identically but yield the conditions $d_1=\bar d_1$ and $d_2=\bar d_2$. 
By taking them into account, we get special cases of the class (III2ii) and (III4ii), respectively, which are illustrated in Fig.\ \ref{fig8}.

\subsection{Pyramidal limit of family (C)}

Here we get again exactly the same conditions as in the Bennett case but for $k=0$. We end up with  special cases of the class (I2), which is illustrated in Fig.\  \ref{fig9}a.

Finally, it should be noted that the prismatic limits and the pyramidal one of family (C) also have the properties of Theorem \ref{thm:property2}.

\begin{figure}[t]
\begin{center}
\begin{overpic}
   [height=45mm]{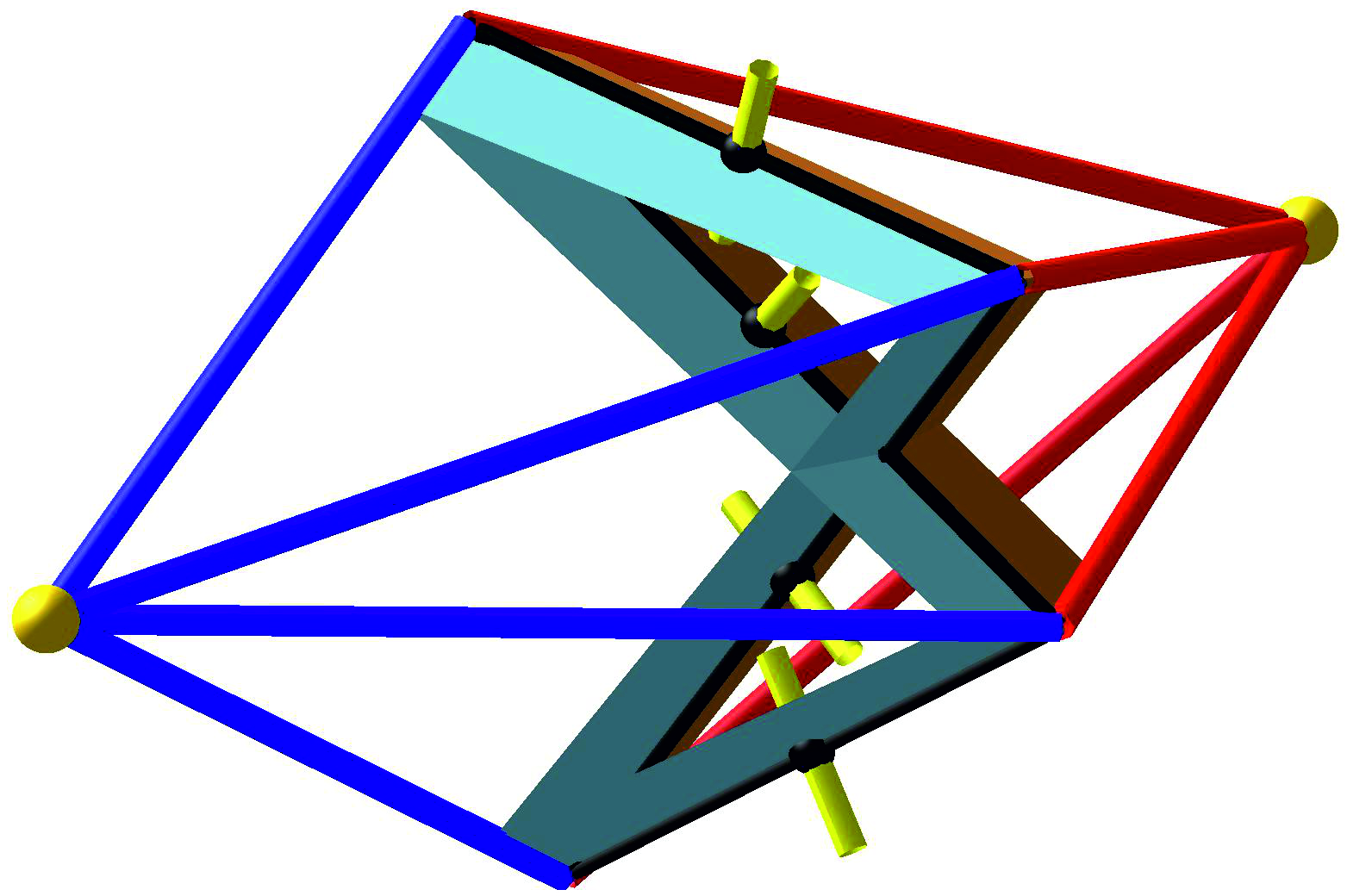}
\begin{scriptsize}
\put(0,0){a)}
\end{scriptsize}     
  \end{overpic} 
\hfill
\begin{overpic}
[height=45mm]{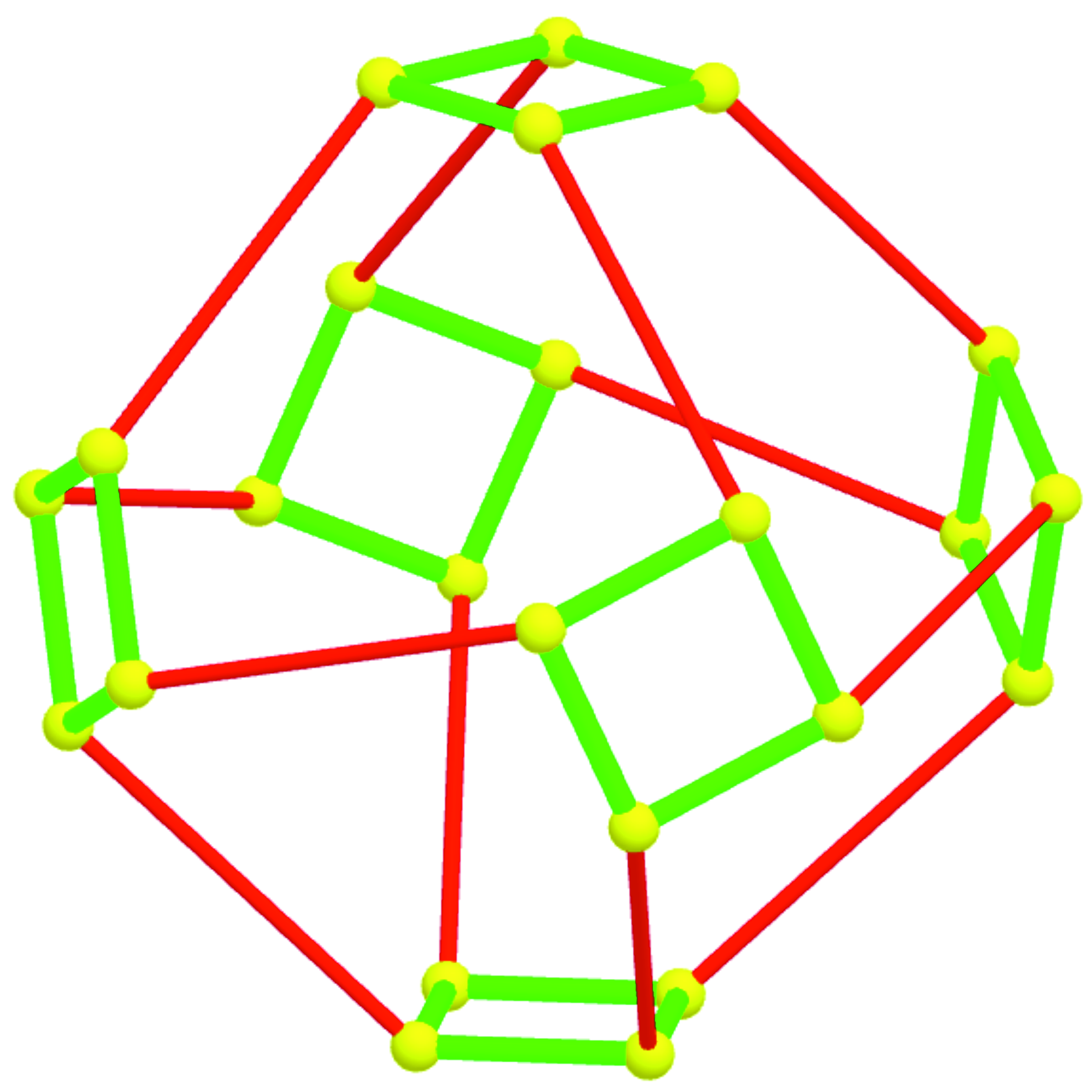}
\begin{scriptsize}
\put(0,0){b)}
\end{scriptsize}  
\end{overpic}
\end{center}	
\caption{(a) Illustration of the pyramidal limit of family (C). 
This illustration corresponds to the following parameters:  
$k=0$, $s=1$ $\mu_{2,3}=\tfrac{2}{3}$, $\mu_{3,4}=\tfrac{1}{2}$, $a_1=\tfrac{1}{2}$, $a_2=\tfrac{1}{3}$, $\tau=\tfrac{3}{4}$ and 
$\bar\tau=-\tfrac{\sqrt{20732639}}{3281}$. (b) Schematic sketch of a truncated octahedron-like structure which consists of 8 rigid hexagonal skew faces and 6 Bennett loops.}
  \label{fig9}
\end{figure}

%%%%%%%%%%%%%%%%%%%%%%%%%%%%%%%%%%%%%%%%%%%%%%%%%%%%%%%%%%%%%%%%%%%%%%%%%%%%%%%%%%%%%%%%%%%%%%%%%%%%%%%%%%%%%%%%%%%%%%%%%%%%%

\section{Conclusion and open problems} \label{sec:conclusion}

We studied the flexible arrangements of two Bennett mechanisms and obtained  
three 4-dimensional solution families (see Theorems \ref{thm:linsym} and \ref{thm:nonsym}).
The families (A) and (B) are globally line-symmetric (see Theorem \ref{thm:property}) and family (C) has a local line-symmetric property (see Theorem  \ref{thm:property2}).  The relation of these three families to flexible bipyramids and biprisms was discussed in Section \ref{sec:limits_results}, and gives rise to the conjecture that there are more general classes of bi-Bennetts which correspond in the limits to the complete classes (I2,II2,III2) and  (I3,II3,III3), respectively. From these more general classes one could also obtain as certain limits 
flexible arrangements of a Bennett tube with a quadrilateral pyramid/prism. Note that these structures are needed for the transition between flexible tubes with planar and skew faces. 

A further interesting question beside the one formulated in Remark \ref{rem:6R} is the following:
Is it possible to replace every pyramid in an octahedron by a Bennett mechanism? The resulting structure would look like a truncated octehdron 
which consists of 8 rigid hexagonal skew faces and 6 Bennett loops (see Fig.\ \ref{fig9}b). 
In this context, it should be noted that another combinatorial arrangement of six Bennett mechanisms exists, which is known as skew parallelepiped possessing even two degrees of freedom \cite{bennett2} (see also \cite[Section IX]{wunderlich}).

\begin{acknowledgement}
This research was funded in whole or in part by the Austrian Science Fund (FWF) [grant DOI
10.55776/F77]. For open access purposes, the author has applied a CC BY public copyright license to any author accepted manuscript version arising from this submission.
\end{acknowledgement}

\section*{Appendix: Non-existence of a flexible plane-symmetric arrangement of Bennett tubes}

Without loss of generality we can first eliminate the scaling factor by setting $k=1$. Then we consider the four points  $\go P_{i,j}$ of Eq.\ (\ref{eq:pij}) and 
consider their coplanarity condition $C=0$ with 
\begin{equation*}
    C:=\det\begin{pmatrix}
        1 & 1 & 1 & 1 \\
        \Vkt P_{1,4} & \Vkt P_{1,2} & \Vkt P_{2,3} & \Vkt P_{3,4} 
    \end{pmatrix}
\end{equation*}
which is of the form 
\begin{equation*}
(a_1-a_2)^2c_4\tau^4 -a_1a_2(a_1-a_2)c_3\tau^3+c_2\tau^2 -a_1a_2(a_1+a_2) c_1\tau  + (a_1+a_2)^2c_0 =0    
\end{equation*}
where the $c_i$'s abbreviate the coefficients of $\tau^i$, which depend on the unknowns $a_1,a_2,\mu_{1,4},\mu_{1,2},\mu_{2,3},\mu_{3,4}$. 
This equation has to be fulfilled independent of $\tau$ implying the 
conditions $c_4=c_3=c_2=c_1=c_0=0$ as $a_1a_2(a_1-a_2)(a_1+a_2)\neq 0$ has to hold. 
Then $c_0-c_4$ splits up into
\begin{equation*}
    -16a_1a_2(a_1-a_2)(a_1+a_2)(\mu_{1,4}\mu_{2,3} - \mu_{1,2}\mu_{3,4})
\end{equation*}
Without loss of generality, we can assume that one $\mu_{i,j}\neq0$ because otherwise all $P_{i,j}$ equal $F_{i,j}$, which form a skew isogram under the assumption of Convention \ref{rem:assumption}. Therefore, we can set 
$\mu_{3,4}=\mu_{1,4}\mu_{2,3}/\mu_{1,2}$.
Then we consider the numerators of the expressions $c_1-c_3$ and $c_1+c_3$, respectively, which factor into:
\begin{equation*}
\begin{split}
    &16a_2(\mu_{1,2} + \mu_{2,3})(\mu_{1,4} - \mu_{1,2})f_1 \\
&16a_1(\mu_{1,2} - \mu_{2,3})(\mu_{1,4} + \mu_{1,2})f_2
\end{split}    
\end{equation*}
with
\begin{equation*}
\begin{split}
&f_1:=a_1^2a_2^2\mu_{1,4}\mu_{2,3} + a_1^2a_2^2 + a_1^2\mu_{1,4}\mu_{2,3} + a_2^2\mu_{1,4}\mu_{2,3} + 3a_1^2 - a_2^2 + \mu_{1,4}\mu_{2,3} + 1 \\
&f_2:=a_1^2a_2^2\mu_{1,4}\mu_{2,3} + a_1^2a_2^2 + a_1^2\mu_{1,4}\mu_{2,3} + a_2^2\mu_{1,4}\mu_{2,3} + 3a_2^2 - a_1^2 + \mu_{1,4}\mu_{2,3} + 1
\end{split}    
\end{equation*}
We have to distinguish the following cases:
\begin{enumerate}
\item 
$f_1=f_2=0$ cannot hold as $f_1-f_2$ yields $4(a_1^2-a_2^2)$, which cannot vanish without contradiction. 
\item 
$\mu_{1,4} = \mu_{1,2} = \mu_{2,3}$: In this case all ${\mu_{i,j}}$ are identical to a value $\mu$. Now $c_4$ factors into $4a_1^2a_2^2\mu^2(a_1^2 - a_2^2)$, which cannot vanish without contradiction.
\item 
$\mu_{1,4} = \mu_{1,2}$ and $f_2=0$. Without loss of generality, we can solve $f_2=0$ for $\mu_{2,3}$. Then the resultant of $c_0$ and $c_2$ with respect to $\mu_{1,2}$ factors into:
\begin{equation*}
    2^{36}a_1^{16}a_2^8
(a_1^2 + 1)^8(a_2^2 + 1)^4(a_1^2 - a_2^2)^4g_1^4g_2^4g_3^4
\end{equation*}
with
\begin{equation*}
g_1:=a_1^2a_2^2 + 2a_2^2 + 1,\quad
g_2:=a_1^2a_2^2 - a_1^2 + 2a_2^2, \quad
g_3:=a_1^2a_2^2 - a_1^2 + 3a_2^2 + 1
\end{equation*}
$g_1=0$ does not have a real solution. 
From $g_2=0$ we get the solution $a_2=a_1/\sqrt{a_1^2+2}$. Back-substitution into $c_0$ shows that there is not real solution for $\mu_{1,2}$. 
From $g_3=0$ we get the solution $a_2=\sqrt{(a_1^2 + 3)(a_1^2 - 1)}/(a_1^2 + 3)$. Back-substitution into $c_0$ shows that it cannot vanish without contradiction.
\item 
$\mu_{2,3} = \mu_{1,2}$ and $f_1=0$. Can be done similarly to item 3 and one also ends up without a real solution. This closes the discussion of all cases.
\end{enumerate}

\begin{thebibliography}{99.}

% alphabetic order

\bibitem{baker_bennett}
J.E.\ Baker: The Bennett, Goldberg and Myard Linkages -- in Perspective. Mechanism and Machine Theory \textbf{14}(4):239--253 (1979)

\bibitem{baker6R}
J.E.\ Baker: A comparative survey of te Bennett-based, 6-revolute kinematic loops. Mechanism and Machine Theory \textbf{28}(1):83--96 (1993)

\bibitem{bennett}
G.T.\ Bennett:
A new mechanism.
Engineering \textbf{76}:777--778 (1903)

\bibitem{bennett2}
G.T.\ Bennett:
The Skew Isogram Mechanism. 
Proceedings of the London Mathematical Society  \textbf{s2-13}(1):151--173 (1914)

\bibitem{bottema}
O.\ Bottema, B.\ Roth: Theoretical Kinematics. North-Holland Publishing Company (1979)

\bibitem{bricard}
R.\ Bricard:  M{\'e}moire sur la th{\'e}orie de l’octa{\`e}dre articul{\'e}. 
Journal de Math{\'e}matiques pures et appliqu{\'e}es, Liouville \textbf{3}:113--148 (1897)

\bibitem{cayley}
A.\ Cayley: Note on the Tetrahedron.
Oxford, Cambridge and Dublin Messenger of Mathematics \textbf{III}:8--10 (1866)

\bibitem{delassus}
E.\ Delassus: 
Les cha{\^i}nes articul{\'e}es ferm{\'e}es et d{\'e}formables {\`a} quatre membres. Bulletin des Sciences Math{\'e}matiques \textbf{46}(2):283--304 (1922)

\bibitem{dietmaier}
P.\ Dietmaier: Einfach \"ubergeschlossene Mechanismen mit Drehgelenken. Habilitation thesis, TU Graz (1995)

\bibitem{hamann}
M.\ Hamann: Line-symmetric motions with respect to reguli. 
Mechanism and Machine Theory \textbf{46}:960--974 (2011)

\bibitem{herve}
J.M.\ Herve, M.\ Dahan: The two kinds of Bennett's mechanisms. 
Proceedings of the Sixth World Congress on Theory of Machines and Mechanisms (J.S.\ Rao and K.N.\ Gupta eds.), pages 116--119, 
J.\ Wiley, New York (1984)

\bibitem{hon}
Y.\ Hon-Cheung: 
The Bennett linkage, its associated tetrahedron and the hyperboloid of its axes. 
Mechanism and Machine Theory \textbf{16}(2):105--114 (1981)

\bibitem{hunt}
K.H.\ Hunt: Kinematic Geometry of Mechanisms. Clarendon Press, Oxford (1978)

\bibitem{kilian}
M.\ Kilian, G.\ Nawratil, M.\ Raffaelli, A.\ Rasoulzadeh, K.\ Sharifmoghaddam: Interactive design of discrete Voss nets and simulation of their rigid foldings. Computer Aided Geometric Design \textbf{111}:102346 (2024)

\bibitem{kokotsakis}
A.\ Kokotsakis: \"Uber bewegliche Polyeder. Mathematische Annalen \textbf{107}: 627--647 (1932)

\bibitem{krames}
J.\ Krames:
Zur Geometrie des Bennett'schen Mechanismus (\"Uber symmetrische Schrotungen V).
Sitz.ber.\ Österr.\ Akad.\ Wiss.\ Math.-Nat.wiss.\ Kl., II, \textbf{146}:159--173 (1937)

\bibitem{naw1}
G.\ Nawratil:  Self-motions of TSSM manipulators with two parallel rotary axes. ASME Journal of Mechanisms and Robotics \textbf{3}(3):031007 (2011) 

\bibitem{naw2}
G.\ Nawratil:  Flexible octahedra in the projective extension of the Euclidean 3-space. Journal for Geometry and Graphics \textbf{14}(2):147--169 (2010)

\bibitem{naw3}
G.\ Nawratil, H.\ Stachel: Composition of spherical four-bar-mechanisms. New Trends in Mechanisms Science – Analysis and Design (D.\ Pisla et al. eds.), pages 99--106, Springer (2010)

\bibitem{nelson}
G.D.\ Nelson: Extending Bricard Octahedra. arXiv:1011.5193 (2010)

\bibitem{pottmann}
H.\ Pottmann, J.\ Wallner: 
Computational Line Geometry. Springer (2001)

\bibitem{kiumars}
K.\ Sharifmoghaddam, R.\ Maleczek, G.\ Nawratil: Generalizing rigid-foldable tubular structures of T-hedral type. Mechanics Research Communications \textbf{132}:104151 (2023) 

\bibitem{song}
C.Y.\ Song, Y.\ Chen, I-M.\ Chen: 
A 6R linkage reconfigurable between the line-symmetric Bricard linkage and the Bennett linkage.
Mechanism and Machine Theory \textbf{70}:278--292 (2013) 

\bibitem{stachel}
H.\ Stachel:  Remarks on Bricard's flexible octahedra of type 3. 
Proceedings of 10th International Conference on Geometry and Graphics, 
pages 8--12 (2002) 

\bibitem{stachel2}
H.\ Stachel
A kinematic approach to Kokotsakis meshes.
Computer Aided Geometric Design 
\textbf{27}(6):428--437 (2010)

\bibitem{horns}
T.\ Tachi:
Designing rigidly foldable horns using Bricard’s octahedron.
ASME Journal of Mechanisms and Robotics \textbf{8}(3):031008 (2016) 

\bibitem{wunderlich}
W.\ Wunderlich: Zur Differenzengeometrie der Fl\"achen konstanter negativer Krümmung. Sitzungsbericht \"Osterr. Akad. Wiss., Math.-Naturw. Kl., Abt. II \textbf{160}:39--77 (1951) 


%%%%%%%%%%%%%%%%%%%%%%%%%%%%%%%%%%%%%%%%%%%%%%%%%%%%%%%




\end{thebibliography}
\end{document}